\documentclass{article}
\usepackage[left=1in, right=1in, top=1in, bottom=1in]{geometry}
\usepackage[utf8]{inputenc}
\usepackage{authblk}
\usepackage{subfig}

\usepackage{amsthm}
\usepackage{amssymb}
\usepackage{amsmath}
\usepackage{mathtools}
\usepackage{stmaryrd}
\usepackage{thm-restate}

\usepackage{csquotes}
\usepackage{enumitem}   
\usepackage{tikz}
\usepackage{graphicx}
\usepackage{xspace}
\usetikzlibrary{arrows.meta}
\usepackage[colorlinks=true,allcolors=blue,allbordercolors=white]{hyperref} % hyperlink setup

\usepackage{caption}  % caption formatting package
\usepackage{dirtytalk}
\usepackage{thm-restate}

\setlist[enumerate]{itemsep=\smallskipamount,parsep=0pt,label={\rm \roman*)}}
\setlist[itemize]{itemsep=\smallskipamount,parsep=0pt}
\newcommand{\defn}[1]{{\textit{\textbf{\boldmath #1}}}\xspace}
\renewcommand{\paragraph}[1]{\vspace{0.09in}\noindent{\bf \boldmath #1.}}

\usepackage{bbm}

\newcommand{\interior}[1]{ {\kern0pt#1}^{\mathrm{o}} }

\usepackage[sanserif,full]{complexity}

\newcommand{\Z}{\mathbb{Z}}

\usepackage{algorithm}
\usepackage[noend]{algpseudocode} % adding [noend] deletes the end while and stuff

\newcommand{\paren}[1]{\left( #1 \right)}

\usepackage[capitalise,nameinlink,noabbrev]{cleveref}
\crefname{equation}{}{} % cref{eq:blah} only does (1) instead of Equation (1)
\crefname{enumi}{Step}{} % cref{eq:blah} only does (1) instead of Item(1)

\theoremstyle{definition}

\newtheorem{definition}{Definition}

\newtheorem{lemma}{Lemma}

\newtheorem{corollary}{Corollary}
\newtheorem{theorem}{Theorem}

\usepackage{mdframed}
\usepackage{framed}

\newcommand{\eps}{\varepsilon}
\newcommand{\softo}[1]{\widetilde{O}\left( #1 \right)}

\usepackage{cryptocode}

\def\draft{1}
\newcommand{\nnote}[1]{\ifnum\draft=1\textcolor{purple}{[\textbf{Nathan:} #1]}\fi}
\newcommand{\znote}[1]{\ifnum\draft=1\textcolor{blue}{[\textbf{Zoe:} #1]}\fi}
\usepackage[
    backend=biber,
    style=alphabetic,
    maxnames=10,
]{biblatex}
\usepackage[utf8]{inputenc}
\usepackage[T1]{fontenc}
\usepackage{amsmath} % For \sqrt and \frac
\usepackage{tikz}
\usepackage{pgfplots}
\pgfplotsset{compat=1.18} % Use a modern compatibility version

\title{The Limits of Black-Box Reductions for\\ All-Pairs Triangle Detection}

\author{Nathan Sheffield, Virginia Vassilevska Williams\thanks{This work is supported in part by NSF CAREER Award 1651838, NSF Grants CCF-1909429 and CCF- 2129139, BSF grants 2016365 and 2020356, a Google Research Fellowship and a Sloan Research Fellowship.}, and Zoe Xi\\
  \normalsize CSAIL, MIT, Cambridge, MA, USA\\
  \normalsize\texttt{\{shefna, virgi, zoexi\}@mit.edu}}
\date{}

\begin{document}

\maketitle

\begin{abstract}

    Consider any tripartite relation $R\subseteq \Z^3$. The $R$-Triangle problem asks, given a graph with edge weights in $\Z$, whether it contains a triangle whose three edge weights form a triple in $R$. The All-Edge $R$-Triangle problem asks to determine for every edge whether it is contained in a triangle whose edge weights form a triple in $R$. Important examples of these problems studied in graph algorithms include the negative triangle problem (where $(x,y,z)\in R$ if and only if $x+y+z<0$), the zero triangle problem (where $(x,y,z)\in R$ if and only if $x+y+z=0$), the monochromatic triangle problem (where $(x,y,z)\in R$ if and only if $x=y=z$) and the (unweighted) triangle detection problem (where $(x,y,z)\in R$ if and only if $x,y,z=1$). Their All-Edge variants are equivalent to, respectively, APSP, AE-Zero Triangles, AE-Mono Triangles and Boolean Matrix Multiplication (BMM).\\
    
    An important result from fine-grained complexity [Vassilevska W.-Williams'10] is that for any relation $R\subseteq \Z^3$, $R$-Triangle and All-Edge $R$-Triangle are subcubically fine-grained equivalent. It implies equivalences between many problems in graph algorithms. However, one direction of the reduction is not tight. It shows that if $R$-Triangle has an $O(n^{3-\eps})$-time algorithm for some $\eps>0$, then All-Edge $R$-Triangle has an $O(n^{3-\eps/3})$-time algorithm. A major open problem is whether this reduction can be made tighter. 
    This paper provides a strong unconditional negative answer to the above question:
    the reduction of [Vassilevska W.-Williams'10] is optimal for black-box reductions that work for arbitrary $R$. \\

    We provide both positive and negative results about black-box reductions between a variety of $R$-triangle problems. Our positive results yield new reductions between several classes of triangle and matrix problems --- for instance, we demonstrate that an $O(n^{2.53})$-time algorithm for computing the equality or dominance product of two $n\times n$ matrices would imply an improvement on known algorithms for computing boolean $(\min, +)$-product, giving the first conditional lower bound for dominance and equality product. Our negative results can be thought of as barriers against natural fine-grained proof techniques. Besides the result that a tighter equivalence between $R$-Triangle and All-Edge $R$-Triangle is not possible, we also
   show that no appropriately ``black-box'' reductions are capable of demonstrating a subcubic equivalence between triangle counting and binary integer matrix multiplication, or a tight equivalence between boolean matrix multiplication and listing $n^2$ triangles, and more, despite the fact that all of these equivalences are conjectured to hold.

\end{abstract}  

\section{Introduction}
One of the most remarkable equivalences in fine-grained complexity is the subcubic equivalence of All-Pairs Shortest Paths (APSP) and the Negative Triangle (NT) problem. APSP asks for the computation of all $n^2$ shortest paths distances between every pair of nodes in a given $n$-node weighted graph. The NT problem asks whether a given $n$-node weighted graph contains three vertices that form a triangle whose edge weights sum up to a negative number. \\

APSP is a function problem, while NT is a decision problem easily reducible to APSP. Nevertheless, Vassilevska W. and Williams \cite{williams2010subcubic} showed that an $O(n^{3-\eps})$-time algorithm for NT for any $\eps>0$ (a so called ``truly subcubic'' time algorithm) can be converted into an $\tilde{O}(n^{3-\eps/3})$-time algorithm for APSP. Hence, finding a truly subcubic time algorithm for APSP is equivalent to obtaining a truly subcubic time algorithm for the seemingly easier NT problem. \\

A benefit of the reduction is that it is not particular to the APSP problem. The same technique yields a subcubic fine-grained equivalence between any triangle problem and its ``all-edge'' variant, i.e. \cite{williams2010subcubic} shows that the following two problems are subcubically equivalent: Let $R \subseteq \Z^3$ be any relation. The input to both problems is a tripartite graph $G=(V,E)$ where $I,J,K$ are the three parts of the tripartition of $V$, $E\subset (I\times J)\cup (J\times K)\cup (I\times K)$ and the edge weights are $w:E\rightarrow \Z$.
\begin{enumerate}
\item {\bf $R$-Triangle}: Determine whether $G$ contains $i\in I,j\in J,k\in K$ with $(i,j),(j,k),(k,i)\in E$ and $(w(i,j),w(j,k),w(i,k))\in R$.

\item {\bf All-Edge $R$-Triangle}: Determine for every edge $(x,y)\in E$ whether there is a $z\in V$ with $(x,z),(y,z)\in E$ and $(w(x,y),w(y,z),w(x,z))\in R$.
\end{enumerate}

Any algorithm for All-Edge $R$-Triangle solves $R$-triangle in the same time. The main result of \cite{williams2010subcubic} is that if $R$-triangle can be solved in $O(n^{3-\eps})$ time for some $\eps>0$, then All-Edge $R$-triangle can be solved in $O(n^{3-\eps/3})$ time.\\

While this reduction gives a subcubic equivalence, it incurs some loss, and diminishing or eliminating this loss would resolve important open problems. 
In particular, for the unweighted case, the all-edge variant is exactly Boolean Matrix Multiplication (BMM), which is strongly believed to have the same complexity as (unweighted) triangle detection. 
However, the reduction of \cite{williams2010subcubic} is unable to show this: even if one were to obtain an optimal, $\tilde{O}(n^2)$-time algorithm for triangle detection, the reduction would only imply an $\tilde{O}(n^{2+2/3})$-time algorithm for BMM, and this is much slower than the known $\tilde{O}(n^\omega)\leq O(n^{2.372})$-time algorithms \cite{alman2404more}. 
It does, though, seem quite plausible, not only that (unweighted) triangle and All-Edge (unweighted) triangle are exactly equivalent, but that in fact  $R$-triangle and All-Edge $R$-triangle are exactly equivalent for \emph{every} $R$. Indeed, all known algorithms for $R$-triangle give algorithms for All-Edge $R$-triangle ``for free''. Hence, we have the following question:

\begin{center}
{\em Question 1: Is a tighter equivalence possible between $R$-Triangle Detection and All-Edge $R$-Triangle Detection, for every fixed $R$?}
\end{center}

One viewpoint one might have with regards to Question 1 is as follows: over the past 15 years, nobody has found {\em any relation $R$} such that there is evidence that all-edge detection is harder than detection, but also nobody has managed to come up with an exact reduction between these problems (even for particular $R$ of interest, such as the unweighted case). So perhaps it is the case that these problems \emph{do} always require the same runtime, but this fact just \emph{isn't} witnessed by any fine-grained reduction.
As reductions are essentially the only way we know to establish equivalence in complexity between problems, this would be a strong barrier to our ability to prove such a fact, even if it is true.\\

Of course, in order for this viewpoint to be coherent, one must place some stipulations on what it means for this fact to be ``provable via fine-grained reduction''. If these two problems have the same complexity, then there is always the trivial reduction between them where one solves all-edge detection directly. To make any sort of meaningful statement, we need to consider a restricted form of reduction --- ideally one that captures the standard techniques used in known fine-grained reductions, while not admitting the trivial reductions of the form ``solve the problem from scratch in the optimal runtime''. 
The definition we propose is as follows:

\begin{definition}\label{def:blackbox}
    For two runtimes $T_A$ and $T_B$ in terms of a vertex count $n$ and edge count $m$, an $R$-triangle problem $A$ (e.g. detection or all-edge detection) $(T_B, T_A)$-\defn{black-box reduces} to another
    $R$-triangle problem $B$ if there exists an algorithm $M$ running
    in $\softo{T_A(n, m)}$ time and making queries to graphs $G_1, \dots, G_r$,
    for $\sum_i T_B(|V(G_i)|, |E(G_i)|) \leq T_A(n, m)$, to an oracle for $B$, such that
    for any relation $R$, $M$ successfully computes $A$ when instantiated
    with a valid $B$ oracle. 
\end{definition}

This is the same as the standard definition of a fine-grained reduction, with the essential restriction that $M$ is \emph{not allowed to depend on $R$}. That is, we would like a single reduction that works when instantiated for \emph{any} relation: even if problems $A$ and $B$ are not efficiently solvable for that relation, if instantiated with a black box solving $B$, the reduction will successfully solve $A$. This is the type of reduction given by \cite{williams2010subcubic}.\\
One of our main results is that, when restricted to black-box reductions, Question 1 has an unconditional negative answer!
That is, the reduction in \cite{williams2010subcubic} is {\bf optimal} for black-box reductions.
Indeed, we prove similar barrier results for a variety of fine-grained complexity problems related to all-edge detection.
In each case, we consider a pair of problems conjectured to require the same runtime for every relation $R$, and show that no single reduction can simultaneously prove this equivalence for all $R$. 
In some cases, while ruling out a tight reduction, we do nonetheless give new quantitative improvements over known reductions.
A full statement of these results follows in \cref{sec:finegrained}.\\

Most of our technical content can be viewed as bounding the cost of determining properties of $3$-uniform hypergraphs in various restricted query models.
In \cref{sec:hypergraphs} we discuss this perspective, and state generalizations of our information-theoretic results which hold for every uniformity of hypergraph.
Then, in \cref{sec:overview}, we give a brief technical overview of the proofs, and outline the structure of the remainder of the paper.

\section{Main results for fine-grained complexity}\label{sec:finegrained}

\subsection{Reducing all-edge detection to detection}

Recall our Question 1 from the introduction: can we get a better quantitative reduction than that of \cite{williams2010subcubic} from all-edge $R$-triangle detection to $R$-triangle detection? Our results rule out any improvement for black-box reductions:

\begin{theorem}\label{cor:cantbeatvirgi}
    All-edge triangle detection does not $(n^{3-\eps}, n^{3-\delta})$-black box reduce to triangle detection (or even triangle counting) for any $\eps>0$, $\delta>\eps/3$.
\end{theorem}

Note that, even if one only cared about these problems with respect to the trivial relation --- that is, one was only interested in reducing Boolean matrix multiplication to (unweighted) triangle detection --- \cref{cor:cantbeatvirgi} could be quite informative.
Our existing techniques for fine-grained triangle reductions are almost all\footnote{There are a number of known reductions between an $R$-triangle problem and an $R'$-triangle problem, for two \emph{distinct} relations $R$ and $R'$. Several of these make important use of the structure of $R$ and $R'$. We are not aware, however, of any nontrivial non-black-box reductions between a pair of \textit{unweighted} triangle problems.} ``black-box'', i.e. agnostic to any choice of underlying relation. 
But \cref{cor:cantbeatvirgi} shows that no techniques of this form will be capable of demonstrating a tight reduction from BMM to triangle detection.
This can be thought of as a barrier explaining why we've been unable to move beyond the bounds of \cite{williams2010subcubic}: while it may be the case that BMM and triangle detection have the same complexity, proving this will require a different approach. To have a chance to improve upon the \cite{williams2010subcubic} reduction for any particular $R$, one must non-trivially use the structure of $R$.

\subsection{Reducing all-edge counting to counting}

Related to the question of reducing BMM to triangle detection, a similarly important open problem in fine-grained complexity is whether fast algorithms for triangle \emph{counting} generically imply fast algorithms for \emph{integer} matrix multiplication.\\

Multiplying $n\times n$ matrices with integer entries bounded by some polynomial in $n$, can be reduced to $O(\log^2 n)$ calls to multiplying matrices with entries in $\{0,1\}$, i.e. binary matrices\footnote{This is easy to see: consider the matrix entries bit by bit and create a matrix product for every pair of bit positions.}. This latter problem is equivalent to the All-Edge Triangle counting problem: for every edge in a graph, determine the number of triangles it is in. Thus the question of whether integer matrix multiplication can be reduced to triangle counting becomes: {\em Can All-Edge Triangle Counting be reduced in a fine-grained way to Triangle Counting?}\\

It is well-known that in the arithmetic circuit model, these two problems are essentially equivalent \cite{Pan78,Pan80}: Any arithmetic circuit for computing the trace of the cube of a matrix over any ring (e.g. the integers) can be converted into a circuit of the same size (up to a constant factor) for multiplying matrices over the ring. This relationship, however, is specific to algebraic algorithms and it is highly non-black-box. (Also, the best known way of converting a word-RAM algorithm into an arithmetic circuit has quadratic overhead, so a tight equivalence for arithmetic circuits doesn't imply such an equivalence for word-RAM algorithms~\cite{StockmeyerVishkin1984}.)  A tight black-box reduction from matrix multiplication (or equivalently, from All-Edge Triangle Counting) to Triangle Counting remains elusive.\\

We know that both integer matrix multiplication and triangle counting are solvable directly in $O(n^\omega) \ll O(n^3)$ time, so there is a trivial subcubic reduction between the problems: just solve them. Thus, we really mean the following:

\begin{center}
{\em Question 2: Is there a (combinatorial) fine-grained subcubic equivalence between triangle counting and all-edge triangle counting, analogous to the one for triangle detection?}
\end{center}

The crucial part of this question is whether there exists a ``simple'' (i.e. combinatorial) reduction, which would be interesting because we do not know ``simple'' truly subcubic time algorithms for either problem. There has been much interest over the years in the answer to Question 2.\\

A priori, a simple reduction seems plausible --- one might for instance try to determine the counts of triangles in some specially chosen subgraphs of the graph defined by the matrix and do some linear algebra to combine them into counts for every pair of vertices.
Our results demonstrate that no reduction of that form is possible, at least without exploiting some graph-theoretic property that wouldn't hold for other relations:

\begin{theorem}\label{cor:no-imm}
    All-edge $R$-triangle counting does not $(n^2, n^{3-\eps})$-black box reduce to $R$-triangle counting for any $\eps > 0$.
\end{theorem}

As with the previous example, our philosophy is to view this as a barrier to fine-grained reductions, in the same vein as e.g. relativization barriers in classical complexity theory. While a resolution of Question 2 may be possible, our ``black-box'' barrier identifies large classes of techniques which may have initially seemed promising but we now know are doomed to fail. We hope this can help hone future efforts to more productive paths, or at least offer some justification as to why past efforts have been fruitless.

\subsection{Reducing listing to all-edge detection}

Another known example of a complexity relationship for triangle problems comes from a pair of reductions by Duraj, Kleiner, Polak, and Vassilevska Williams~\cite{duraj2020equivalences} between the all-edge triangle detection problem and the problem of listing $m$ triangles (or all triangles if there are fewer than $m$) in an $m$-edge graph.
Their reduction shows that these problems are equivalent up to polylog factors when one considers runtimes in terms of the number of edges. An examination of their reductions reveals that they hold in our ``black-box'' sense for any relation $R$:

\begin{theorem}[Duraj, Kleiner, Polak, and Vassilevska Williams~\cite{duraj2020equivalences}]\label{thm:listing-to-alledge}
    For any $1 \leq c \leq 1.5$, all-edge triangle detection $(m^c, m^c)$-black-box reduces to listing $m$ triangles, and listing $m$ triangles $(m^c, m^c)$-black-box reduces to all-edge triangle detection.
\end{theorem}

One might wonder whether a similar result can be shown for runtimes in terms of the number of vertices.

\begin{center}
{\em Question 3: Is the problem of listing $n^2$ triangles equivalent to all-edge detection for runtimes in terms of $n$?}
\end{center}

In one direction, their argument applies directly:

\begin{theorem}[Duraj, Kleiner, Polak, and Vassilevska Williams~\cite{duraj2020equivalences}]\label{thm:alledge-to-listing}
    For any $2 \leq c \leq 3$, all-edge triangle detection $(n^c, n^c)$-black-box reduces to listing $n^2$ triangles.
\end{theorem}

However, in the other direction, such a translation is unclear. The reverse direction of this implication would be quite interesting, because for dense graphs often we \emph{can} get fast algorithms for all-edge $R$-triangle detection using matrix products, so this might be an even more directly useful statement than \cref{thm:listing-to-alledge}. 
We give a quantitatively weaker version of a positive answer to Question 3: 

\begin{theorem}\label{cor:ub-triangle-all}
    Assuming $\omega = 2$ (i.e. that integer matrix multiplication is solvable in $\softo{n^2}$ time), for any $2 \leq c \leq 3$, the problem of listing $n^2$ triangles $\paren{n^{c}, n^{\paren{\frac{5c-9 + \sqrt{297-162c+25c^2}}{4}}}}$-black-box reduces to all-edge triangle detection.
\end{theorem}

This is a subcubic reduction\footnote{Although this is perhaps not easy to tell from looking at the expression --- see \cref{fig:tikzplot} later in the paper for a plot of this exponent against $c$.} between these problems for any relation $R$, but not a full equivalence as exists in the sparse case.
In fact, we demonstrate that a full equivalence between these problems, if true, won't follow from existing methods.
We are unable to prove a full black box separation, but we rule out reductions with a further restriction:

\begin{definition}
    Say a reduction between weighted graph problems is \emph{non-duplicitous} if no edge weight ever appears with higher multiplicity in a call of the reduction than it did in the input graph.
\end{definition}

\begin{restatable}{theorem}{duplicity}\label{cor:lb-triangle-all}
    For any $2 \leq c \leq 3$ and any $\eps > 0$, there is no non-duplicitous $\paren{n^c, n^{\paren{\frac{9}{6-c} - \eps}}}$-black-box reduction from the problem of listing $n^2$ triangles to all-edge triangle detection.
\end{restatable}

This non-duplicitousness constraint holds, for instance, for any reduction which only makes queries to \emph{subgraphs} of the original graph.
All known relationships between triangle problems which hold for general $R$ (in particular, those of \cite{williams2010subcubic} and \cite{duraj2020equivalences}) are shown via such subgraph reductions. So, this is evidence that new techniques are needed if we want a positive answer to Question 3.

\subsection{Relating the complexity of ``rotated'' matrix products}

Recall that many forms of matrix multiplication can be thought of as all-edge $R$-triangle detection problems for corresponding relations $R$.
Although we think of these problems in tripartite graphs, most of the matrix products we care about give rise to relations $R$ that are symmetric in the three parts $I$, $J$, and $K$. For instance, Boolean matrix multiplication is equivalent to all-edge (unweighted) triangle detection, and APSP is equivalent to all-edge negative triangle detection, where $R = \{(w_{ij}, w_{jk}, w_{ki}) \ | \ w_{ij} + w_{jk} + w_{ki} < 0\}$. \\

However, one can also consider asymmetric relations $R$, where the edge weights between e.g. parts $I$ and $J$ are treated differently than those between parts $J$ and $K$. In these cases, the problem of detecting which edges \emph{between $I$ and $J$} belong to $R$-triangles may be a qualitatively quite different problem than detecting which edges \emph{between $J$ and $K$} belong to $R$ triangles.\\

There are, in fact, some examples of asymmetric $R$ for which these problems are important.
For instance, consider $R = \{w_{ij}, w_{jk}, w_{ki} \ |\ w_{jk} \geq w_{ki}, w_{ij}=1\}$. 
In this case, the problem of all-edge $R$-triangle detection between parts $I$ and $J$ 
can be readily seen to be equivalent to computing the so called {\em dominance product} \footnote{Technically, this is the so called existence version of dominance product. The original definition of the dominance product of $A$ and $B$ is $C_{ij}=|\{k~|~A_{ik}\leq B_{kj}\}|$. Under this definition, the dominance product is known to be runtime-equivalent (see \cite{labib2019hamming,vnotes}) to (counting version of) the equality product given by $C_{ij}=|\{k~|~A_{ik}= B_{kj}\}|$.}  of matrices $A$ and $B$ 
 %the dominance product of matrices $A$ and $B$ is 
defined as the Boolean matrix $C$
whose $(i, j)$ entry is $1$ if and only if there is some $k$ such that
$A[i, k] \leq B[k, j]$ 
(see e.g. \cite{matouvsek91,maxweighttri,fredmandom}). \\

On the other hand, consider the problem of all-edge $R$-triangle detection between parts $J$ and $K$ for the same choice of $R= \{w_{ij}, w_{jk}, w_{ki} \ |\ w_{jk} \geq w_{ki} , w_{ij}=1\}$. Via a simultaneous binary search argument one can show that this problem is equivalent (up to log factors) to the \defn{Boolean-$(\min,
  +)$-product} --- that is, the problem of computing $(\min, +)$-product between an $n\times n$
integer-valued matrix and an $n\times n$ $\{0, \infty\}$-valued
matrix, which is important for the all-pairs shortest
paths problem in node-weighted graphs~\cite{chan2007more,AbboudF0WX25}.

Both dominance product \cite{matouvsek91} and
Boolean-$(\min, +)$-product \cite{chan2007more} are known to be solvable in time
$\widetilde{O}(n^{(\omega + 3)/2}) \leq O(n^{2.69})$, where $\omega <
2.372$ is the matrix multiplication constant \cite{alman2404more}, via
essentially the same algorithmic techniques; however, neither problem
is known to reduce to the other.
In fact, the dominance product 
%and the provably at least as hard equality product \cite{labib2019hamming,vnotes,fredmandom} are 
is among the very few {\em intermediate} problems (those with known running time between $n^\omega$ and $n^3$) without {\em any} known conditional lower bounds\footnote{Beyond the obvious fact that it's at least as hard as BMM so it likely requires $n^{\omega-o(1)}$ time. However, if $\omega=2$ this is not meaningful and we would really like a higher than $n^{\omega-o(1)}$ time conditional lower bound.}.
Meanwhile, the Boolean-$(\min, +)$-product is known to require at least $n^{2.5-o(1)}$ time under the APSP Hypothesis \cite{nickfischer2026} (see also \cite{fredmandom,ChanWX21}). Any nontrivial reduction between the rotated versions of the corresponding triangle problem would therefore be very interesting because it would give a hardness result for dominance product, which has been thus far elusive.\\

One might conjecture that, in fact, ``rotated'' $R$-triangle detection problems are \emph{always} equivalent:

\begin{center}
{\em Question 4: Does Boolean-$(\min, +)$-product reduce to dominance product? More generally, is all-$IJ$ $R$-triangle detection equivalent to all-$JK$ detection for every $R$?}
\end{center}

Indeed, we know of no $R$ for which we have reason to suspect different complexities. However, no reduction of this form is known.\\

In addition to \cref{cor:ub-triangle-all}, we give a (quantitatively weaker) reduction from listing triangles to \emph{detection just between a single pair of parts} ---
which, along with \cref{thm:alledge-to-listing}, implies the following:

\begin{theorem}\label{cor:rotate-abitt}
    For any $2 \leq c \leq 3$, all-$IJ$ detection $(n^c, n^{\paren{\frac{9-2c}{4-c}}})$-black-box reduces to all-$JK$ detection. So, an $n^c$-time algorithm for dominance product yields an $n^{\paren{\frac{9-2c}{4-c}}}$-time algorithm for Boolean-$(\min, +)$-product.
\end{theorem}

\cref{cor:rotate-abitt} is the first nontrivial reduction between dominance product and Boolean-$(\min, +)$ product. 
The reduction is unfortunately not strong enough to give a superquadratic lower bound for dominance product, independent of $\omega$. However, it does show that if Boolean-$(\min, +)$-product requires $n^{2.5+\delta-o(1)}$ time for some constant $\delta>0$, then dominance product requires $n^{2+\eps-o(1)}$ time for $\eps=4\delta/(1+\delta)>0$. Thus, as long as Boolean-$(\min, +)$-product requires at least $\Omega(n^{2.6115})$ time (this is lower than the current best known running time), dominance product requires at least $\Omega(n^{2.372})$ time, and this is higher than the current best matrix multiplication running time.
The reduction also implies that computing dominance product in $O(n^{2.53})$ time would suffice to improve upon the current best known algorithm for Boolean-$(\min, +)$-product.\\

Getting a better reduction seems difficult: interpreting matrix multiplication variants as all-edge $R$-triangle detection problems, we show that \cref{cor:rotate-abitt} is tight for non-duplicitous black-box reductions.

\section{Main results for hypergraph query complexity}\label{sec:hypergraphs}

Our upper and lower bounds for black-box reductions all come from thinking about the problems information-theoretically.
If the relation $R$ is unknown, these $R$-triangle problems on a graph $G$ can be thought of as \emph{hyperedge detection} (resp. listing, counting, etc.) problems on an unknown hypergraph: the vertices of the hypergraph are the vertices of $G$ and the hyperedges are the triples of vertices whose edge triples satisfy $R$. \\

Observe that, if the edge weights of the graph are all distinct, by choice of the relation $R$ one can make it so that any arbitrary collection of triples of vertices correspond to $R$-triangles. So, if one wishes to use a reduction to a problem $B$ to solve an $R$-triangle problem $A$ in a black-box sense, one is effectively attempting to answer a question about an \emph{arbitrary} unknown $3$-uniform hypergraph, where access to this hypergraph is provided in the form of queries to $B$.\\

This perspective has already proved quite fruitful in generating generally-applicable positive results in fine-grained complexity. Prior work has observed that the question of ``what information about the set of relation-satisfying triples can we determine via black-box use of a detection algorithm'' can be thought of in terms of query complexity in the \defn{independent set oracle} model. That is, the query model in which access to an unknown hypergraph is provided via an oracle which takes in a set of vertices and determines whether or not they contain at least one hyperedge. One known result is that, for a $k$-uniform, $k$-partite hypergraph, it is possible to approximately count the number of hyperedges by making only $\polylog(n)$ many independent set queries --- in particular, this gives a black-box reduction from approximate counting to detection for triangle problems, as well as higher-arity relations like $k$-SUM or $k$-CLIQUE~\cite{beame2020edge, dell2021fine, bhattacharya2019hyperedge, dell2022approximately,Censor-HillelEW25}.

\subsection{ All-\texorpdfstring{$s$}{s} codegree and all-\texorpdfstring{$s$}{s} containment queries}

In this paper, we consider a new type of hypergraph query in which, instead of just learning whether some subset of the $3$-uniform hypergraph contains at least one hyperedge, we learn \emph{for every pair of vertices} whether that pair belongs to a hyperedge in the queried subset. This is the exactly the kind of information we would get from an all-edge detection algorithm. In fact, we could consider a similar definition for higher uniformity graphs and larger collections of vertices --- one could imagine an oracle which takes in a subset of an $r$-uniform hypergraph and answers ``for every collection of $s$ vertices, are those $s$ vertices contained in at least one hyperedge together''? Although the case $r=3$, $s=2$ is the one relevant for our $R$-triangle applications, the algorithms and lower bounds we present generalize to these higher-dimensional analogues. 

\newcommand{\Rreplace}{U} % zoe: in case we want to change it again

\begin{definition}
    For a fixed integer $s\geq 1$, an \defn{all-$s$ containment oracle} for an unknown hypergraph $H = (V, E)$ takes as input a pair of subsets $Q \subseteq V$ and $\Rreplace \subseteq \mathcal{P}(Q)$ (where $\mathcal{P}$ denotes the power set), and returns the list of all $s$-element subsets of $Q$ that are contained in at least one hyperedge in $\Rreplace \cap E$. 
\end{definition}

Note that this definition allows queries to specify not only a subset of vertices, but also a subset of \emph{candidate hyperedges} $\Rreplace$ to intersect the hyperedge-set of the hypergraph with.
We will primarily discuss the strictly weaker version of this model in which we always require $\Rreplace = \mathcal{P}(Q)$, which we call the \defn{induced} version of the query model.
(Note that the standard independent set query model can be thought of as an ``induced all-$0$ containment oracle''.)\\

We also define the following related model, a strengthened version of hyperedge \emph{counting} queries (also known as \defn{additive queries} in the literature):

\begin{definition}
    For a fixed integer $s\geq 1$, an \defn{all-$s$ codegree oracle} for an unknown hypergraph $H = (V, E)$ takes as input a pair of subsets $Q \subseteq V$ and $\Rreplace \subseteq \mathcal{P}(Q)$, and returns the codegree of (i.e. number of hyperedges containing) each of the $\binom{|Q|}{s}$ distinct sets of $s$ vertices in the hypergraph $(Q, \Rreplace \cap E)$.
\end{definition}

We will generally work with $r$-partite, $r$-uniform hypergraphs. 
In that setting, we could also weaken the oracles by letting them only return information about the $s$-tuples of vertices spanning a specific set of parts. (This is important for our discussion of ``rotated'' matrix products, where we wanted to reduce $R$-triangle listing to all-edge detection restricted only to the edges between a particular pair of parts --- see \cref{sec:finegrained}.) 

\begin{definition}
    For an $r$-partite hypergraph $H =  (X_1 \sqcup \dots \sqcup X_r, E)$, a \defn{specified-parts} all-$s$ codegree oracle takes as input a pair of subsets $Q\subseteq V$ and $\Rreplace\subseteq\mathcal{P}(Q)$ and returns the codegree of every tuple $(v_1, \dots, v_s) \in (Q \cap X_1) \times \dots \times (Q \cap X_s)$ in the subhypergraph $(Q, \Rreplace \cap E)$. (Correspondingly, a specified-parts all-$s$ containment oracle returns whether each of these values is $0$.)
\end{definition}

As our motivation for studying algorithms in these query models is to instantiate the oracles with efficient algorithms and get an efficient algorithm out, our results will deal with oracles whose costs scale with the size of their input, as opposed to unit cost queries.
When we say an oracle has ``cost $n^c$'' for some $c$, this means that an algorithm making queries $(Q_1, \Rreplace_1), \dots, (Q_q, \Rreplace_q)$ is defined to have query complexity $\sum_{1 \leq i \leq q} |Q_i|^c$.
Our cost functions will always be of the form $n^c$ for some $s \leq c \leq r$, where $r$ is the uniformity of the hypergraph, as an all-$s$ algorithm with cost less than $n^s$ is impossible due to the output size, and an all-$s$ algorithm with cost $O(n^r)$ is trivial so long as we can efficiently decide whether a given $r$-tuple is an edge.\\

In addition to relating these query models to previously studied ones, the primary problem we will focus on solving with these oracles is the \defn{exact recovery} problem: given query access to an unknown hypergraph, can we learn a complete explicit description of the hypergraph?
%Clearly, the query complexity of this problem is at least as hard as that of any other problem.
For general $r$-uniform hypergraphs, it is reasonably straightforward to see that no nontrivial algorithm can solve this: a full description of the hypergraph is simply too much information to obtain with few queries.
However, if we restrict the hypergraph to have a small number of hyperedges, this is an interesting problem that has been studied previously in a number of different query models. (And, indeed, this problem turns out to be exactly what we need for our results on the fine-grained complexity of triangle listing.)

\subsection{Results}

Note that, for all of our results, we treat parameters $s$, $c$, and $r$ as constants, so that our asymptotic notation is asymptotic in terms of the vertex and edge counts. We write tildes over $O$, $\Omega$, $o$, and $\omega$ to denote hiding of $\polylog$ factors in addition to constants.\\

We begin in \cref{sec:ezstories} by comparing our all-$s$ codegree and all-$s$ containment oracles with the previously studied independent set and additive (i.e. hyperedge counting) oracles.
We show that (the not-necessarily-induced version of) independent set queries can simulate all-$s$ containment queries at some quantitative loss, via a direct generalization of Vassilevska Williams and Williams's \cite{williams2010subcubic} reduction from All-Edge $R$-Triangles to $R$-Triangle detection.
However, we show that this quantitative loss is necessary and tight, even if we instead have access to additive queries.
We also observe that an all-$s$ containment oracle can do no better than the trivial approach for simulating additive queries (meaning also that it cannot simulate an all-$s$ codegree oracle).
\begin{theorem}\label{thm:sims}
    Given any $s < c < r$, for unknown $r$-partite, $r$-uniform hypergraphs we have the following:
    \begin{enumerate}
        \item Given a cost-$n^c$ oracle for detecting whether a subhypergraph contains an edge, all-$s$ containment queries can be answered at cost $\softo{n^{\paren{s + \frac{c(r-s)}{r}}}}$. 
        \item Given a cost-$n^c$ oracle for counting the number of edges in a subhypergraph, $\widetilde{\Omega}\left(n^{\paren{s + \frac{c(r-s)}{r}}}\right)$ cost is required to answer an induced all-$s$ containment query.
        \item Given a cost-$n^c$ oracle for counting the number of edges in a subhypergraph, $\widetilde{\Omega}(n^{\min(r, s+c)})$ cost is required to answer an induced all-$s$ codegree query.
        \item Given a cost-$n^c$ all-$s$ containment oracle, $\widetilde{\Omega}(n^{r})$ cost is required to count the total number of hyperedges.
    \end{enumerate}
\end{theorem}

We then begin considering the problem of exact recovery using these oracles. In \cref{sec:specified-parts-upperbound}, we give our main algorithmic result (with the special case of $r=3$, $s=2$ proved in \cref{sec:triangles}):

\begin{restatable}{theorem}{specifiedpartsub}\label{thm:specifiedpartsub}
    Given cost-$n^c$ specified-parts induced all-$s$ containment oracle access to an unknown $r$-partite, $r$-uniform hypergraph $H = (V, E)$, we can exactly recover $H$ with high probability at cost 
    $\softo{|V|^{r(1-\gamma)}|E|^{\gamma}},$
    where $\gamma = \frac{r-c}{2r-s-c}$.

\end{restatable}

Then, in \cref{sec:lbs}, we show lower bounds:

\begin{restatable}{theorem}{specifiedpartslb}\label{thm:specifiedparts-lb}
    For any values of $s\leq c <r$, and any sufficiently large values of $|V|$ and $|E| < \frac{1}{2}\binom{|V|/r}{r}$, any algorithm that exactly recovers unknown $r$-partite, $r$-uniform hypergraphs $H = (V, E)$ with high probability using cost-$n^c$ specified-parts all-$s$ codegree queries must incur cost 
    $\widetilde{\Omega}\paren{|V|^{r(1-\gamma)}|E|^{\gamma}}$
    on some instances, where $\gamma = \max\paren{\frac{r-c}{2r-s-c}, \frac{r-c}{r-1}}$.
\end{restatable}

\begin{restatable}{theorem}{allpartslb}\label{thm:allparts-lb}
    For any values of $r$, $s<r$, and any sufficiently large values of $|V|$ and $|E| < \frac{1}{2}\binom{|V|/r}{r}$, any algorithm that exactly recovers unknown $r$-partite, $r$-uniform hypergraphs $H = (V, E)$ with high probability using cost-$n^c$ all-$s$ codegree queries must incur cost 
    $\widetilde{\Omega}\paren{|V|^{r(1-\gamma)}|E|^{\gamma}}$
    on some instances, where $\gamma = \max\paren{\frac{\binom{r}{s}(r-c)}{\binom{r}{s}(r-s) + (r-c)}, 1 - \frac{c-1}{(r-s) + (s-1)\binom{r}{s}}}$.
\end{restatable}

Observe that the upper bound of \cref{thm:specifiedpartsub} is tight with the lower bound of \cref{thm:specifiedparts-lb} in our motivating case of $r = 3$, $s=2$. (In fact, it is tight whenever $c \geq r - s + 1$.) However, it remains far from the lower bound of \cref{thm:allparts-lb} in the case where we get this containment information for every pair of parts. While we do not know how to match \cref{thm:allparts-lb} even in the case of $r=3$, $s=2$, we do demonstrate that this all-parts information can grant additional power, showing the following improvement over \cref{thm:specifiedpartsub}:

\begin{restatable}{theorem}{allpartsub}\label{thm:allparts-ub}
    Given cost-$n^c$ all-pairs containment oracle access to an unknown $3$-partite, $3$-uniform hypergraph $H = (V, E)$, we can exactly recover $H$ with high probability at cost 
    \[\softo{|V|^{\frac{-1 + \sqrt{9 + 4\eps^2}}{2} - \eps}|E| + |V|^{\frac{3 + \sqrt{9 + 4\eps^2}}{2} - \eps}},\]
    where $\eps = 3-c$.
\end{restatable}

We note that, while our other results are fully algorithmic in addition to information theoretic (that is, the computational overhead of the algorithm is always upper bounded by the cost of its queries), this algorithm further requires listing many triangles in an explicitly-known graph, which requires no query cost but introduces an additional computational overhead not accounted for in this query cost. We next give an upper bound that also considers computational costs. See \cref{fig:tikzplot} for an illustration of these bounds in the case that $|E| = |V|^2$. 

\begin{restatable}{theorem}{allpartsubcomp}\label{thm:allparts-ub-comp}
    Assuming $\omega = 2$ (i.e. that integer matrix multiplication is solvable in $\softo{n^2}$ time), given cost-$n^c$ all-pairs containment oracle access to an unknown $3$-partite, $3$-uniform hypergraph $H = (V, E)$, we can exactly recover $H$ with high probability at computational cost 
    \[\softo{|V|^{\frac{-2-5\eps + \sqrt{36+12\eps+25\eps^2}}{4}}|E| + |V|^{\frac{6-5\eps + \sqrt{36+12\eps+25\eps^2}}{4}}},\]
    where $\eps = 3-c$.
\end{restatable}

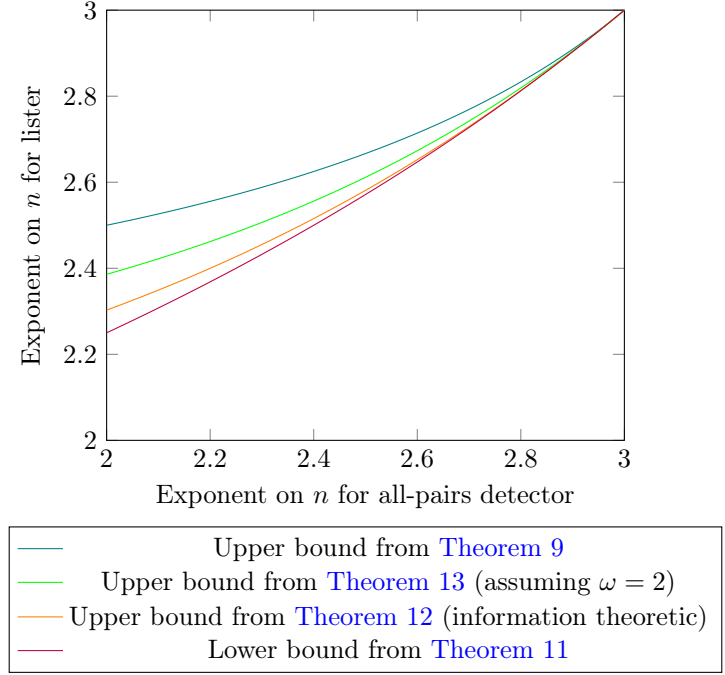
\begin{figure}[h]
    \centering
\begin{tikzpicture}
\begin{axis}[
    xmin=2, xmax=3,
    ymin=2, ymax=3,
    domain=2:3,
    samples=500,
    xlabel={Exponent on $n$ for all-pairs detector},
    ylabel={Exponent on $n$ for lister},
    legend style={at={(0.5,-0.20)},anchor=north}
]

\addplot[teal]
{ 3 - (3 - x)/(4 - x) };
\addlegendentry{Upper bound from \cref{thm:specifiedpartsub}}

\addplot[green]
{ (6 - 5*(3-x) + sqrt(36 + 12*(3-x) + 25*(3-x)^2))/4 };
\addlegendentry{Upper bound from \cref{thm:allparts-ub-comp} (assuming $\omega=2$)}

\addplot[orange]
{ (3 + sqrt(9 + 4*(3 - x)^2))/2 - (3 - x) };
\addlegendentry{Upper bound from \cref{thm:allparts-ub} (information theoretic)}

\addplot[purple]
{ 3 - 3*(3 - x)/(6 - x) };%{ (21 - 6*x)/(10 - 3*x) };
\addlegendentry{Lower bound from \cref{thm:allparts-lb}}

\end{axis}
\end{tikzpicture}

    \caption{The cost of exactly recovering an $n$-vertex, $n^2$-edge, $3$-uniform hypergraph, in terms of the cost of an all-pairs detection oracle. The teal line is unconditional, the green line requires standard integer matrix multiplication to be solvable in time $\softo{n^2}$, and the orange line considers only query cost and not the cost of auxiliary computation.}
    \label{fig:tikzplot}
\end{figure}

\subsection{Interpretation in terms of black-box fine-grained reductions}

Studying the power of all-$2$ containment oracles in $3$-uniform hypergraphs is a very closely related pursuit to studying the power of black-box reductions to all-edge $R$-triangle detection. 
Hence, our results in \cref{sec:finegrained} can be viewed as \emph{moral} corollaries of these query complexity upper and lower bounds. 
However, there are some incomparable details of the models, meaning that our $R$-triangle results cannot quite be written as direct consequences as the results in this section (nor vice-versa) --- instead, we note where appropriate how to modify the proofs to handle those details.\\

For one, an upper bound on query complexity is applicable in fine-grained reductions only if the query algorithm furthermore incurs very little \emph{computational} overhead.
This is the case for most of our above results, but not for \cref{thm:allparts-ub}, which is why we have to use the quantitatively weaker \cref{thm:allparts-ub-comp} for our fine-grained reductions.\\

For another, a lower bound on all-$2$ containment query cost does not quite suffice to rule out all types of black-box reduction. 
If we had a black-box reduction to all-edge $R$-triangle detection which \emph{only ever called the all-edge $R$-triangle algorithm on subgraphs of the original graph}, then this would directly correspond to an algorithm in the all-$2$ containment query model for unknown hypergraphs\footnote{Because, as noted, if all edge weights are distinct, $R$ can encode an arbitrary hypergraph, and each call made in the black-box reduction can be expressed as a query to the all-$2$ containment oracle. Note that if the black-box reduction call is made to an induced subgraph, this gives an induced all-$2$ containment query, but if the query deletes edges we must remove all triples involving those edges from consideration in the all-$2$ containment query.}.
However, an additional approach such reductions can take is to \emph{rearrange} or \emph{reweight} edges, calling the $R$-triangle detection oracle on other constructed graphs that don't appear as a subgraph of the input. 
Hence, in order to translate our results on all-$2$ containment queries to lower bounds against arbitrary black-box reductions, we have to show that such queries do not afford any additional power beyond the ability to query subgraphs. \\

Finally, note that even if computational overheads are small, upper bounds in the all-$2$ containment query model need not be implementable in $R$-triangle reductions. The reason is that the all-$2$ containment oracle model allows excluding an arbitrary collection of candidate triples from consideration, whereas in the $R$-triangle setting we can only exclude triples from consideration by deleting some composite edge. We are largely able to avoid worrying about this issue, since most of our upper bounds hold in the \emph{induced} all-$2$ containment model.

\section{Technical overview}\label{sec:overview}

\subsection{Proof ideas}

Here, let us briefly outline the most important technical ideas in our main results.\\

\textbf{Query complexity lower bounds.} At a high level, our lower bounds on query cost in the unknown hypergraph model take the following structure. We fix a distribution over hypergraphs such that the distribution over the correct output of the algorithm has high entropy. Then, we consider some relaxation of the query model --- i.e. some nicely-behaved random variables which provide strictly more information than the query responses. We show that no deterministic (adaptive) algorithm with low total query cost can cause that collection of random variables to have sufficiently high mutual information with the distribution of correct outputs. By Yao's principle, this implies that even randomized algorithms must fail with high probability.\\

For the case of computing all-$s$ codegree using hyperedge-counting queries, our distribution is simply a uniform random hypergraph (i.e. each hyperedge present with independent probability $1/2$). Here, we define a relaxation of the query model which ensures that the set of hyperedges involving each $(r-s)$-tuple of vertices in the last $(r-s)$ parts of the hypergraph all remain independent random variables conditional on the results of any number of queries. Then we show that, for any low-cost query algorithm, some such $(r-s)$-tuple must have high conditional entropy even given all of the query results.\\

For the case of computing all-$s$ containment, we instead consider a random hypergraph with exactly $n^s$ edges. We relax the queries to be ``thresholded'', in that whenever the conditional distribution on hypergraphs given all previous queries places sufficiently high probability on a given hyperedge being present, we reveal whether that hyperedge is present. This allows us to argue that small queries always reveal very little information, because they are very unlikely to involve any unrevealed hyperedges. Bounding the information returned from larger queries by $O(\log(n))$ bits, we show that the optimal size and number of queries matches the reduction of \cite{williams2010subcubic}. \\

For lower bounds on solving exact recovery with all-$s$ codegree queries, we consider a slightly more complicated distribution. We first fix a small number of random ``admissible'' $s$-tuples of vertices for each $s$-tuple of vertex parts. Then, we place each hyperedge uniformly at random conditional on intersecting each $s$-tuple of vertex parts in an ``admissible'' $s$-tuple. This produces more tightly-clustered hypergraphs (which is necessary because exact recovery is easy in the Erd\H{o}s--R\'enyi case). Our analysis looks similar to the all-$s$ containment to hyperedge-counting lower bound, except that we run it in parallel for the neighbourhood of every admissible $s$-tuple, and we further use the fact that no large query can contain a high density of admissible $s$-tuples.\\

\textbf{Lower bounds for black-box $R$-triangle reductions.} These query cost lower bounds in the unknown hypergraph model immediately translate into lower bounds against black-box \emph{subgraph} reductions computing all-edge $R$-triangle counting, all-edge $R$-triangle detection, and $R$-triangle listing, respectively. It is straightforward to modify the first two lower bound proofs to rule out arbitrary black-box reductions (i.e. where calls to the reduction can be arbitrary weighted graphs, not necessarily subgraphs of the input). However, relaxing the subgraph condition in the listing case requires more work, since the lower bound crucially used the subgraph condition to argue that no query contained too many admissible pairs. Using a combinatorial lemma about Ruzsa--Szemer\'edi-type induced packings, we are able to argue that no arrangement of weighted edges can contain many candidate triangles without placing a large fraction of edges in triangles already known to be relation-satisfying (a detection query reveals no further information for those edges). This allows us to weaken the subgraph restriction to the more mild restriction that no edge weight appears with higher multiplicity in a call of the reduction than in the original graph. (That is, we permit reductions to rearrange edges, just not duplicate them.) \\

\textbf{Upper bounds for exact recovery.} We'll describe the case $r=3$, $s=2$, as this is both the simplest nontrivial case and the one most relevant to our black-box reduction results. The basic idea is as follows: we first use random sampling and our all-pairs detection oracle to compute estimates of the codegree of each pair of vertices. For a pair whose codegree is sufficiently high, we make $O(n)$ constant-size queries to determine all hyperedges involving that pair. To handle pairs of lower codegrees, we randomly partition the hypergraph into vertex subsets such that, for every sufficiently low-codegree pair of vertices, for every hyperedge involving them, there is constant probability that this is the only hyperedge involving them in its triple of vertex subsets. Then, for each triple of vertex subsets, we run an appropriate ``parallel binary search'' to recover the identity of each such unique hyperedge. (A similar argument works for higher uniformities, although there are several complications that arise. For instance, we can no longer afford to partition finely enough to fully isolate hyperedges in their tuple of vertex subsets --- instead, we use the partitioning to find in parallel a collection of high-codegree vertex subsets, and argue that every hyperedge will be found with constant probability by brute-forcing over all $r$-tuples involving those subsets.)\\

This algorithm yields a tight reduction in the case that our all-pairs oracle only returns results between a single pair of vertex parts. However, we show that it is possible to substantially improve on this approach given a full all-pairs oracle.
We handle low-codegree pairs in a similar way.
But now, instead of handling each high-codegree pair separately, we make use of the graph structure on the high-codegree pairs.
To handle the vertices involved in a large number of these high-codegree pairs, we make several recursive calls to the algorithm on subhypergraphs.
For the remaining vertices, we observe that the graph defined by the high-codegree pairs must be highly clustered, in the sense that it contains a much larger number of triangles than a random graph of the same max-degree would.
Such graphs are highly structured --- in particular, one can capture a large number of the triangles in a small collection of vertex subsets~\cite{gupta2014decompositions}.
We use these ideas to iteratively process all high-codegree pairs in batches at less cost than it would take to brute-force over each of them individually.

\subsection{Structure of the paper}

In \cref{sec:ezstories}, we prove \cref{thm:sims}, which implies our lower bounds on reducing all-edge $R$-triangle detection and all-edge $R$-triangle counting to $R$-triangle counting. Then, in \cref{sec:triangles}, we give upper bounds for the exact recovery problem in the special case of $r=3$, $s=2$, which is the one relevant for triangle detection. In \cref{sec:hypergraphs}, we generalize those results to hypergraphs of arbitrary uniformity. Then, in \cref{sec:lbs}, we show lower bounds on exact recovery (in generality). Finally, in \cref{sec:duplicitousness}, we show how to translate these exact recovery lower bounds into lower bounds on black-box reductions from $R$-triangle listing to all-edge $R$-triangle detection.

\subsection{Background and related work}
There is a rich literature surrounding learning (hyper)graphs, or determining properties thereof, from restricted query models.
Much of this focuses on independent set queries, which determine whether a given subset of vertices contains a hyperedge.
There are lines of work on optimal parameters for learning graphs~\cite{alon2004learning,alon2005learning,angluin2005learning,chang2014learning,abasi2019learning} and hypergraphs~\cite{angluin2005learning,d2016multistage,abasi2018error,abasi2018non,chang2018learning,balkanski2022learning,austhof2025non} using these queries.
Another particularly well-studied problem in the independent set query model is estimating the number of edges~\cite{beame2020edge,chen2020nearly,dell2021fine} or hyperedges~\cite{bhattacharya2019hyperedge,dell2022approximately,dell2024nearly,Censor-HillelEW25}. The latter reference \cite{Censor-HillelEW25} considers independent set queries with a non-constant cost.
Note that in some of these works there is a separation between what can be done in bipartite graphs (or $r$-partite $r$-uniform hypergraphs) compared to general graphs.\\

Graph and hypergraph recovery problems have also been studied in a number of other global query models, including queries that count the number of edges~\cite{choi2008optimal,bshouty2009reconstructing,bshouty2010optimal}, measure distances~\cite{mathieu2013graph}, test connectivity~\cite{kluk2024graph}, or find maximal independent sets~\cite{konrad2025graph}.

\section{Preliminaries}

For $n\in\mathbb{N}$, we write $[n]$ for the set $\{1, 2, \ldots, n\}$.\\

Our lower bound results make use of Yao's lemma, which gives a way to lower bound the worst-case performance of a randomized algorithm.

\begin{lemma}[Yao's lemma, \cite{yao1977probabilistic}]\label{lem:yao}
    Let $\mathcal{X}$ be a finite set of inputs, and let $\mathcal{A}$ be a finite set of deterministic algorithms that solve a problem on the inputs in $\mathcal{X}$. Let us use $\mathcal{R}$ to denote a probability distribution on $\mathcal{A}$, and $\mathcal{D}$ to denote a probability distribution on $\mathcal{X}$. Let $c(A, x)$ denote the cost of running algorithm $A\in\mathcal{A}$ on input $x\in\mathcal{X}$. We have the following:
    \[    \min_{\mathcal{R}}\max_{x\in\mathcal{X}}\mathbb{E}_{A\leftarrow\mathcal{R}}[c(A, x)] = \max_{\mathcal{D}}\min_{A\in\mathcal{A}}\mathbb{E}_{x\leftarrow\mathcal{D}}[c(A, x)].
    \]
\end{lemma}

Our results come from thinking about hyperedge detection problems on an unknown hypergraph. A hypergraph is a generalization of a graph where every edge is allowed to contain any number of vertices. An $r$-uniform hypergraph is a hypergraph where every edge contains exactly $r$ vertices. We will mostly consider $r$-partite $r$-uniform hypergraphs, which are $r$-uniform hypergraphs $H$ on $r$ parts of vertices such that every edge contains exactly one vertex from each part, i.e., we can write $H = (V = V_1\sqcup\cdots\sqcup V_r, E)$, where $V_1, \ldots, V_r$ are the parts of $V$ and we have $E\subseteq V_1\times\cdots\times V_r$. For a set $S \subseteq V$, we define the codegree of $S$ to be $|\{ e \in E \ | \ S \subseteq E\}|$ --- this is a generalization of the notion of degree in graphs.

\section{Relating counting and detection to their all-\texorpdfstring{$s$}{s} versions}\label{sec:ezstories}

The all-$s$ containment and all-$s$ codegree queries we discuss in this paper can be viewed as direct strengthenings of two query models that have been extensively previously studied: \emph{independent set queries}, which detect whether an induced subhypergraph contains at least one hyperedge, and \emph{additive queries}, which count the number of edges in an induced subhypergraph.
Our new models differ by additionally providing versions of this information for \emph{every} $s$-tuple of vertices in the query.
So, an all-$s$ containment (resp. all-$s$ codegree) oracle at cost $n^c$ immediately yields an independent set (resp. additive) oracle at cost $n^c$.
Similarly, additive (resp. all-$s$ codegree) queries at cost $n^c$ directly yield independent set (resp. all-$s$ containment) queries at cost $n^c$, as we can just check which counts are $0$.
In this section, we investigate other relationships between these query models. \\

As a positive result, we observe that it \emph{is} possible to nontrivially simulate all-$s$ containment queries using an independent set oracle that's allowed to query \emph{subhypergraphs} as opposed to just \emph{induced subhypergraphs}.
In other words, to decide if a given collection of $r$-tuples of vertices on a particular vertex set includes a hyperedge.
This model has been called \defn{edge emptiness queries} or \defn{OR queries} in the literature~\cite{assadi2021graph,bishnu2023complexity}.
We have the following (note that, as in the case of our all-$s$ problems, we measure the cost of an oracle call in terms of the number of vertices it involves as opposed to the number of possible edges):

\begin{lemma}\label{lem:use-is}
    Given edge emptiness queries at cost $n^c$ to an unknown $r$-partite $r$-uniform hypergraph, we can compute the answer to an all-$s$ containment query at cost $\softo{n^{\paren{s + \frac{c(r-s)}{r}}}}$.
\end{lemma}
\begin{proof}
    We wish to solve the all-$s$ containment problem on an $n$-vertex $r$-partite $r$-uniform hypergraph.
    It suffices to solve the specified-parts all-$s$ containment problem (i.e. only report containment for vertices spanning the first $s$ parts), as a generic algorithm for this could then be run $\binom{r}{s}$ times to determine values for every $s$-tuple of parts.
    We first arbitrarily partition each of the $r$ parts into $n^{s/r}$ distinct chunks, each of size at most $n^{\paren{\frac{r-s}{r}}}$.
    Then, we iterate through every $r$-tuple of those chunks, and do the following: 
    Make an edge emptiness query to the subhypergraph induced by the union of those chunks. If the oracle reports that a hyperedge exists, then find such a hyperedge by binary search (by repeatedly throwing away half of the vertices in a given part and testing if a hyperedge still remains, we can find some hyperedge with $O(\log n)$ edge emptiness queries). 
    Then, consider the $s$-tuple of vertices that the hyperedge involves among the first $s$ parts.
    Record that this $s$-tuple \emph{is} involved in some hyperedge, and subsequently exclude all sets containing those $s$ vertices together from all queries made in the rest of the reduction (note that this is where we use that we are able to query any subhypergraph, not just induced subhypergraphs).
    Repeat this whole process again until the edge emptiness oracle returns that there are no more hyperedges between these chunks, at which point we progress to the next $r$-tuple of chunks.\\

    To argue correctness, we note that for any $s$-tuple of vertices among the first $s$ parts, if it is involved in a hyperedge, then some such hyperedge will be found: by the time we get to the first $r$-tuple of chunks involving such a hyperedge, the only hyperedges that we will be excluding from consideration in the induced subhypergraph will be those involving a different $s$-tuple of vertices among the first $s$ parts.
    To argue cost, note that every time we attempt to find a hyperedge, we either successfully mark some previously-unmarked $s$-tuple as belonging to a hyperedge (which can happen at most $n^s$ times), or we move on to the next $r$-tuple of chunks (which can happen at most $(n^{s/r})^r = n^s$ times).
    The cost of running that ``attempt to find a hyperedge'' procedure is at most $\log n$ many edge emptiness queries to a subhypergraph on at most $n^{\frac{r-s}{r}}$ vertices, and so has cost $O\paren{\paren{n^{\paren{\frac{r-s}{r}}}}^c}$.
    Overall, this gives cost $\softo{n^{\paren{s + \frac{c(r-s)}{r}}}}$. 
\end{proof}

In the case of $r = 3$, $s = 2$, note that this proof exactly corresponds to Vassilevska Williams and Williams's reduction from All-Edge $R$-Triangle Detection to $R$-Triangle detection.
The fact that their reduction is implementable in the hypergraph edge emptiness model suggests a limitation of their techniques for proving a stronger equivalence.
Specifically, we can show information theoretically that the above is tight even given the \emph{counting} version of edge emptiness queries. We start with the following observation:

\begin{lemma}\label{lem:find-uniquers}
    Given a cost-$n^c$ all-$s$ containment oracle, we can at cost $\softo{n^c}$ exactly recover all hyperedges that involve some $s$-tuple with codegree at most $\polylog(n)$.
\end{lemma}
\begin{proof}

We first give an algorithm for recovering, for each $s$-tuple in a unique hyperedge, the hyperedge that the $s$-tuple is in, and then bootstrap this to recover all hyperedges involving an $s$-tuple with codegree at most $\polylog(n)$.\\

Write the hypergraph as $H = (V = X_1\sqcup, \ldots, \sqcup X_r), E)$. We give an algorithm for recovering all such $s$-tuples in the specific tuple of parts $(X_1, \ldots, X_s)$ (then, to get an algorithm for recovering all such $s$-tuples, we can just run this algorithm on all possible $s$-tuples of parts). The idea is to do binary search on all the vertex parts to simultaneously find for every $s$-tuple $x$ in exactly one hyperedge, the hyperedge that it is in. For each $i\in\{s+1, \ldots, r\}$, let $f_i: X_i\to [|X_i|]$ be an
ordering on the vertices in $X_i$.  Initialize a list $L$ of size $n^s$ where each
entry is a tuple of $r-s$ empty strings. For every
$x$ in a unique hyperedge $e$, the $i$th string in $L[x]$ will contain
$f_i(v_i)$ where $v_i$ is the vertex in $X_i$ in $e$.\\

For $j\in[\log n]$, we split each $X_i$ into $X_{i, 0}$ and $X_{i,
  1}$, where $X_{i, 0}$ (resp. $X_{i, 1}$) consists of all vertices
$v_i\in X_i$ such that $f_i(v)[j]$ is $0$ (resp. $1$). For $k\in
[2^{r-s}]$, we run $\mathcal{A}$ on $X_1\sqcup\cdots\sqcup X_s\sqcup
X_{s+1, k[1]}\sqcup\cdots\sqcup X_{r, k[r-s]}$ (where $k[\ell]$ is the
$\ell$th bit of the $(r-s)$-length binary representation of $k$). For
each $x\in V(s)$, if $\mathcal{A}$ outputs that $x$ is in some
hyperedge, then for each $\ell\in [r-s]$ we append $k[\ell]$ to
$L[x][\ell]$. At the end, for every $s$-tuple $x$ in a unique hyperedge, $L[x]$ contains the other vertices in the hyperedge with $x$. We return $L$. This algorithm runs in time $\widetilde{O}(n^c)$ since we make $\widetilde{O}(1)$ cost-$n^c$ calls to the oracle.\\

Now we bootstrap this algorithm: to do this, the idea is to, in $K\polylog(n)$ iterations where $K$ is some large enough constant, partition each of the vertex parts of $H$ into equal-sized chunks of size $Cn/\polylog(n)$ for some large enough $C$, then run the algorithm previously described on all $r$-tuples of chunks, one from each part, and output all the hyperedges that the algorithm lists (in any iteration). Applying a Chernoff bound and choosing $C$ large enough, it can be readily seen that with high probability every $s$-tuple that has codegree at most $\polylog(n)$ is listed. The cost is $\widetilde{O}(n^c)$ since we make $\widetilde{O}(1)$ calls to the cost-$n^c$ oracle.

\end{proof}

We can now show the following:

\begin{lemma}\label{lem:cant-use-is}
    Fix some algorithm which makes queries to a cost-$n^c$ oracle for counting the number of edges in subhypergraphs of an unknown $r$-partite $r$-uniform hypergraph.
    If this algorithm solves the all-$s$ containment problem with high probability, then it requires cost $\widetilde{\Omega}\paren{n^{\paren{s + \frac{c(r-s)}{r}}}}$ for some hidden hypergraphs.
\end{lemma}
\begin{proof}
    Suppose we had a more efficient reduction. By \cref{lem:find-uniquers}, this also gives us a $\widetilde{o}\paren{n^{\paren{s + \frac{c(r-s)}{r}}}}$ cost algorithm to exactly recover an $r$-uniform, $r$-partite hypergraph so long as each $s$-tuple belongs to at most $\polylog(n)$ distinct edges.
    We will show that this is impossible.
    By Yao's principle, it suffices to fix a distribution over hypergraphs such that with high probability each $s$-tuple belongs to at most $\polylog(n)$ hyperedges, and such that no deterministic algorithm with so small a query complexity can succeed with high probability over this distribution.
    The distribution we choose will simply place $n$ vertices in each part and include every hyperedge with independent probability $n^{s - r}$.
    Observe that this distribution has entropy $n^r H(n^{s-r})$ --- so, any algorithm solving the exact recovery problem with high probability must have an output distribution with mutual information at least $(1-o(1))n^r H(n^{s-r})$ with the hypergraph.
    By the data processing inequality, this implies that the distribution of all the results of its queries must have entropy at least $(1 - o(1))n^r H(n^{s-r})$.\\

    The goal now becomes bounding the total amount of information revealed by any given query.
    First, we pass to a relaxation of the oracle model: for any $r$-tuple of vertices, if at any point in the course of the algorithm's runtime the probability of that $r$-tuple being a hyperedge conditional on all previous queries exceeds $2n^{s-r}$, we will reveal for free whether it is involved in a hyperedge.
    We'll let random variable $X_i$ denote the response to the $i$th query, and random variable $Y_i$ denote the information revealed after the $i$th query from these threshold exceedences.
    Observe that this thresholding causes us to provide strictly more information to the algorithm, so if we can rule out algorithms in this query model we rule them out in the original model.
    Observe also that, since we expect at most $n^r/2$ threshold exceedences, the expected amount of information we reveal in response to them is insufficient to capture all of the entropy of the distribution: 
    \[\sum_{i} H(Y_i \ | \ X_1, \dots, X_{i}, Y_1, \dots, Y_{i-1}) \leq \frac{n^r}{2} \cdot H(2n^{s-r}) \leq n^r H(n^{s-r}) - \Omega(n^s).\]
    By the chain rule for conditional entropy, this implies that we have 
    \[\sum_{i} H(X_i \ | \ X_1, \dots, X_{i-1}, Y_1, \dots, Y_{i-1}) \geq \Omega(n^s).\]

    We now need to show that this is impossible for any query strategy that always has small query cost.
    Consider some run of the algorithm in which it made (possibly adaptively chosen) queries $(Q_1, \Rreplace_1), \dots, (Q_q, \Rreplace_q)$ for $\sum_{i=1}^q |Q_i|^c < n^{\paren{s + \frac{c(r-s)}{r}}} \cdot \log^{-2}(n)$,
    and got responses $x_1, \dots, x_q$, with additional bonus information $y_1, \dots, y_q$ revealed.
    We now bound $H(X_i \ | \ X_1 = x_1, \dots, X_{i-1}=x_{i-1}, Y_1 = y_1, \dots, Y_{i-1}= y_{i-1})$ in terms of $|Q_i|$. 
    When $|Q_i| \geq n^{\paren{\frac{r-s}{r}}}$, we use the trivial bound of $H(X_i) \leq O(\log n)$ (which holds because $X_i$ is supported only on integers between $0$ and $n^{r-s}$).
    When $|Q_i| < n^{\paren{\frac{r-s}{r}}}$, we instead use subadditivity of entropy.
    Such a query can involve at most $|Q_i|^r$ distinct $r$-tuples of vertices that haven't yet been revealed to have be a hyperedge, and we know each of them has probability at most $2n^{s-r}$ conditional on our previous queries of being a hyperedge.
    So, the total conditional entropy of this query is at most $|Q_i|^r H(2n^{s-r}) < O(|Q_i|^r n^{s-r}\cdot\log(n))$.
    Together, we can therefore bound (assuming small cost and appealing to convexity in the penultimate inequality): 
    \begin{align*}
    \sum_{i} H(X_i \ | \ X_1 = x_1, \dots, X_{i-1}=x_{i-1}, Y_1=y_{1}, \dots, Y_{i-1}=y_{i-1}) &\leq \\
    \sum_{\substack{i \\ |Q_i|> n^{\paren{\frac{r-s}{r}}}}} O(\log n) + \sum_{\substack{i \\ |Q_i|< n^{\paren{\frac{r-s}{r}}}}} O(|Q_i|^r n^{s-r}\cdot\log(n)) &\leq \\
    \frac{n^{\paren{s + \frac{c(r-s)}{r}}} \cdot \log^{-2}(n)}{\paren{n^{\paren{\frac{r-s}{r}}}}^c} \cdot O(\log(n)) + \frac{n^{\paren{s + \frac{c(r-s)}{r}}} \cdot \log^{-2}(n)}{\paren{n^{\paren{\frac{r-s}{r}}}}^c} \cdot O(\log(n))&\leq\\ O(n^s \cdot \log^{-1}(n)). 
    \end{align*}

    Since we've shown this sum of conditional entropies is $o(n^s)$ for every run with small query complexity, but $\Omega(n^s)$ in expectation, we know that the algorithm must incur query cost at least $n^{\paren{s + \frac{c(r-s)}{r}}} \cdot \log^{-2}(n)$ on some instances.
\end{proof}

Since the all-$s$ codegree problem is strictly harder than the all-$s$ containment problem, \cref{lem:cant-use-is} also gives some lower bound on the cost of solving all-$s$ codegree with additive queries. Our following lemma strengthens that lower bound:

\begin{lemma}\label{lem:cant-allscount}
    Fix some algorithm which makes queries to a cost-$n^c$ oracle for counting the number of edges in subhypergraphs of an unknown $r$-partite $r$-uniform hypergraph.
    If this algorithm solves the all-$s$ codegree problem with high probability, then it requires cost $\widetilde{\Omega}\paren{n^{\min(r, s+c)}}$ for some hidden hypergraphs.
\end{lemma}

Observe that it is trivial to achieve cost $O(n^r)$ (simply iterate over every $r$-tuple of vertices and make a constant-cost query to check if it is a hyperedge), or cost $O(n^{s+c})$ (for every $s$-tuple of vertices, remove all other vertices sharing parts with them and use a single cost-$n^c$ oracle query to determine the number of hyperedges involving those $s$ vertices).
This lemma demonstrates that nothing can substantially improve on those algorithms.

\begin{proof}[Proof of \cref{lem:cant-allscount}]
    We consider a strictly stronger model of query.
    For every $(r-s)$-tuple of vertices spanning the last $(r-s)$ parts of the hypergraph, let the associated ``bucket'' be the set of all $r$-tuples involving those $(r-s)$ vertices.
    Observe that any additive query $(Q, \Rreplace)$ can intersect at most $|Q|^{r-s}$ buckets, and will reveal a $O(\log(n))$-bit value (namely, a single count).
    We could instead imagine providing the algorithm with a stronger oracle: for any $q$, at cost $q^c$, the algorithm can evaluate \emph{any} $O(\log(n))$-bit-valued function on each of \emph{any} collection of $q^{r-s}$ many buckets.
    This oracle model is at least as powerful, as the algorithm could choose a set of buckets corresponding to all $(r-s)$-tuples of the vertices in $Q$, learn for each of them the number of hyperedges shared with $\mathcal{P}(Q) \cap \Rreplace$, and then add up all those values.\\

    Now, consider a hidden $r$-uniform $r$-partite hypergraph chosen by including each hyperedge with independent probability $1/2$.
    The set of hyperedges in each bucket is an independent random variable with $n^s$ bits of entropy.
    Now, consider some (possibly adaptive) strategy of querying our strengthened oracle, and suppose that for some $i$ the $i$th bucket is queried $\widetilde{o}(n^{s})$ times in expectation.
    After any sequence of queries, each bucket's set of hyperedges is still independent from all other buckets.
    So, only the queries we make to this bucket yield any information about its contents. 
    Since each query returns $O(\log n)$ bits per bucket, this means that the results of all of our queries have mutual information $o(n^s)$ with the $i$th bucket's hyperedges, meaning that the conditional entropy of the $i$th bucket's hyperedges is $\Omega(n^s)$ given the results of all queries.
    We claim that this implies that the algorithm must fail the all-$s$ codegree problem with overwhelming probability.
    Indeed, suppose that following the algorithm's last query, we subsequently revealed the contents of every bucket other than the $i$th. 
    Now, determining all-$s$ codegrees entails exactly determining the set of hyperedges in the $i$th bucket, which remain distributed as before as all buckets are independent.
    As conditioning cannot increase conditional entropy, this means that the correct output to the all-$s$ codegree problem has conditional entropy $\Omega(n^s)$ given the results of all of the algorithm's queries.
    By Fano's inequality, any estimator for the output of the all-$s$ codegree problem given those query results must fail with high probability.\\

    Finally, we observe that, since a query of cost $x = q^c$ involves at most $q^{r-s} = x^{(r-s)/c}$ buckets, and queries must have sizes at least $r$ and at most $rn$, by convexity any sequence of queries of cost $\widetilde{o}(n^{\min(r, s+c)})$ must query the average bucket at most $\widetilde{o}(n^s)$ times.
    So, an algorithm whose query cost is always this low must fail with high probability on the uniform distribution.
\end{proof}

The final question we consider is whether all-$s$ containment gives any use in simulating additive queries.
It is again easy to see that, at cost $O(n^{r})$, one can learn the entire hypergraph, just by querying every $r$-tuple of vertices individually.
The following straightforward lower bound shows that no improvement to this is possible:

\begin{lemma}\label{lem:cant-count}
    Fix some algorithm which makes queries to a cost-$n^c$ all-$s$ containment oracle for an unknown $r$-partite $r$-uniform hypergraph, with $s < c < r$.
    If this algorithm returns the number of hyperedges in the hypergraph with high probability, then it requires cost $\Omega\paren{n^{r}}$ for some hidden hypergraphs.
\end{lemma}
\begin{proof}
    Again, we use Yao's principle: it suffices to define a distribution on which any deterministic algorithm with $o(n^r)$ query complexity is likely to err.
    The distribution we choose will with $1/2$ probability be the complete $r$-partite $r$-uniform hypergraph on $n$ vertices, and with $1/2$ probability be that hypergraph with one random hyperedge removed.
    For a given query $(Q, \Rreplace)$, observe that any $s$-tuple of vertices in $Q$ that belongs to at least two distinct sets of size $r$ in $\Rreplace$ will necessarily return true, as at least one of those sets must be an edge.
    So, the query will only provide information distinguishing these two distributions if some $s$-tuple of vertices in $Q$ belongs to a \emph{unique} $r$-tuple in $\Rreplace$, and that $r$-tuple happens to be the missing edge of the graph.
    There can be at most $\binom{|Q|}{s}$ many $r$-tuples tested this way by a single query, so since $\sum_{i} |Q_i|^s \leq \sum_i |Q_i|^c \leq o(n^r)$, an algorithm with $o(n^r)$ query cost will with high probability never make such a distinguishing query.
\end{proof}

\cref{lem:use-is}, \cref{lem:cant-use-is}, \cref{lem:cant-allscount}, and \cref{lem:cant-count} together constitute our \cref{thm:sims}.

\subsection{Interpretation in terms of black-box fine-grained reductions}

As noted, any lower bound on complexity in the $3$-uniform hypergraph query model directly translates to lower bounds on black-box triangle reductions that are \emph{restricted to only ever querying subgraphs}.
So, plugging in $r=3$ and $s=2$, \cref{lem:cant-use-is} and \cref{lem:cant-allscount} directly imply \cref{cor:cantbeatvirgi} and \cref{cor:no-imm} if we consider only reductions that always query subgraphs.\\

It turns out these are the interesting cases --- the full versions of \cref{cor:cantbeatvirgi} and \cref{cor:no-imm} follow directly from slight modifications to those proofs. Specifically, note that in the proof of \cref{lem:cant-use-is}, the only properties we use of the queries are that 
\begin{itemize}
    \item a query of size $q$ involves at most $q^3$ distinct candidate triangles, and
    \item each query returns at most $O(\log(n))$ bits.
\end{itemize}
Both of these facts remain true even if the reduction is allowed to rearrange or duplicate edges. (Any graph on $q$ vertices can contain at most $q^3$ triangles, so we consider at most that many candidates from the original graph. And the $R$-triangle counting query still must return an integer between $0$ and $q^3$.) So, we can conclude \cref{cor:cantbeatvirgi}.\\

For the other case, we need to be slightly more careful.
The proof of \cref{lem:cant-allscount} (when restricted to $r=3$, $s=2$) used the fact that a query of size $q$ can reveal information about only $q$ distinct vertices in part $K$ of the tripartition.
This is true for subgraph queries, but is no longer the case when we allow the reduction to query arbitrary graphs --- it is no longer the case that each vertex in the query corresponds to a vertex in the original graph.
In principle, a query could return information pertaining to up to $q^2$ different vertices in $K$, one for each edge in the query.
Fortunately, when $2= s \leq c \leq r =3$, our convexity argument would still have applied even if we bounded the number of buckets involved in a size-$q$ query by $q^2$ instead of $q$.
That is, it's still the case that any sequence of queries of $\widetilde{o}(n^3)$ cost must query some bucket at most $\widetilde{o}(n^2)$ times, meaning that the reduction fails with high probability.
So, we can conclude \cref{cor:no-imm}.

\section{Upper bounds for exact recovery in the case \texorpdfstring{$r=3, s=2$}{r=3, s=2}}\label{sec:triangles}

In this section, we give algorithms for the problem of
recovering a $3$-uniform, $3$-partite hypergraph given access to a
cost-$n^c$ induced all-pairs containment oracle, in the case when the
oracle is a specified-parts oracle and in the case when the oracle is
an all-parts oracle. These results are particularly relevant for the
fine-grained complexity discussion. (In the language of $R$-triangles,
we give black-box algorithms for the problem of listing
$R$-triangles in a tripartite graph given access to a cost-$n^c$
oracle for all-edge triangle detection.)\\

Here, we use $G = (V = I\cup J\cup K, E)$ for the $3$-uniform,
$3$-partite graph and $\mathcal{A} = \mathcal{A}_G$ for the oracle,
and when $\mathcal{A}$ is an oracle for a specific pair of parts, we
take $I$ and $J$ to be these parts. We use $n$ for $|V|$ as usual. We
write $|E| = t = n^\tau$, where $0\le\tau\le 3$. We write $c =
3-\eps$, for some $0\le\eps\le 1$, so that $\mathcal{A}$ is a
cost-$n^{3-\eps}$ oracle. To start, we give a few subroutines that
both of our algorithms use.\\

The first subroutine uses the oracle to estimate for every pair of
vertices in $I\times J$ the number of hyperedges that it is in.

\begin{lemma}\label{lem:estimate}
 Given $G$ and access to $\mathcal{A}$, we can compute in randomized
 $\widetilde{O}(n^{3-\eps})$ time a list of estimates $a_{i,j}$ for
 every pair of vertices $(i,j) \in I \times J$, such that with high
 probability every pair $(i,j)$ belongs to at least $\frac{a_{i,
     j}}{10}$ and at most $10a_{i, j}$ distinct hyperedges in $G$.
\end{lemma}

\begin{proof}
For each $\ell \in \{ 0, \dots, \log n \}$, for $1000 \log n$
iterations, we sample a uniformly random subset $K' \subseteq K$,
where $|K| = 2^\ell$, and call $\mathcal{A}$ on $I \cup J \cup K'$.
For a particular pair $(i,j) \in I \times J$, let $x$ be the number of
distinct hyperedges in $G$ that contain $(i,j)$, and let $y_\ell$ be
the fraction of these $1000\log n$ iterations for which $\mathcal{A}$
detected a hyperedge containing $(i,j)$.  Applying Chernoff bounds, we obtain the following:
\begin{itemize} 
\item If $x < \frac{2^\ell}{10}$, then $\Pr\left[y_\ell < \frac{1}{4}\right] \geq 1 - \frac{1}{n^3}$.
\item If $x > 10 \cdot 2^\ell$, then $\Pr\left[y_\ell > \frac{3}{4}\right] \geq 1- \frac{1}{n^3}$.
\item If $2^{\ell-1} \leq x \leq 2^\ell$, then $\Pr\left[\frac{1}{4} \leq y_\ell \leq \frac{3}{4}\right] \geq 1- \frac{1}{n^3}$.
\end{itemize}
So, our algorithm can set $a_{i,j}$ to be $2^{\ell}$, where $\ell$ is
the smallest value such that $1/4 \leq y_\ell \leq 3/4$.  Our
guarantees ensure that with high probability such an $a_{i,j}$ will
exist for every pair in $I\times J$, and will be within a factor of
$10$ of the true number of hyperedges containing $(i, j)$.
\end{proof}  

The next subroutine uses the oracle to list, for every pair of
vertices in $I\times J$ that is in a \emph{unique} hyperedge, the
hyperedge that it is in.

\begin{lemma}\label{lem:find-uniquers-2}
Given $G$ and access to $\mathcal{A}$, there exists an
$\widetilde{O}(n^{3-\varepsilon})$-time algorithm that, for every
$(i,j) \in I \times J$ involved in \emph{exactly} one
hyperedge in $G$, lists that hyperedge.
\end{lemma}

\begin{proof}

Let $g: K\to [|K|]$ be an ordering on the vertices in $K$. We run a
binary search on $K$ simultaneously for each $(i, j)\in I\times J$ to
find a vertex $k$ such that $(i, j, k)\in E$ (if one exists).\\

We initialize an array $L$ of empty strings, one for each $(i,j) \in I
\times J$.  For $\ell = 1, \ldots, \log n$, set $K'\subseteq K$ to
consist of the vertices $k\in K$ such that the $\ell$th bit of $g(k)$
is $0$. We call $\mathcal{A}$ on $I\cup J\cup K'$. For each $(i, j)$,
if $\mathcal{A}$ returns that $(i, j)$ is in some hyperedge, we append
$0$ to $L[(i, j)]$; otherwise, we append $1$. Upon exiting the outer
for-loop, for each $(i, j)$, we check whether the vertex $k$ such that
$g(k) = L[(i, j)]$ is in a hyperedge with $i$ and $j$, and if $k$ is,
we output $(i, j, k)$.\\

As we call $\mathcal{A}$ $O(\log(n))$ times on instances of size
$O(n)$, this algorithm runs in time $\widetilde{O}(n^{3-\eps})$.

\end{proof}  

Finally, we give a subroutine that bootstraps the algorithm described
in \cref{lem:find-uniquers-2} to list, for every pair of vertices $(i,
j)\in I\times J$ that's in not-too-many hyperedges, all the hyperedges
that it is in.

\begin{lemma}\label{lem:find-kinda-uniquers}
  Given $G$, access to $\mathcal{A}$, and some parameter $\alpha > 0$, there
  exists a randomized $\widetilde{O}(n^{3-\varepsilon +
    \varepsilon\alpha})$-time algorithm such that the following
  guarantee holds with high probability: for pair of vertices $(i,j)
  \in I \times J$ involved in at most $n^\alpha$ distinct hyperedges
  in $G$, the algorithm lists all hyperedges involving $(i,j)$.
\end{lemma}
\begin{proof}
  Call a pair of vertices $(i, j)\in I\times J$ \defn{unpopular} if it
  is involved in at most $n^{\alpha}$ distinct hyperedges.  Note that
  it suffices to give an $\widetilde{O}(n^{3 - \eps +
    \eps\alpha})$-time algorithm with the guarantee that every
  hyperedge involving an unpopular pair is listed with constant
  probability. This is because repeating $O(\log n)$ independent
  iterations of that algorithm will w.h.p. list all such hyperedges.\\
  
  The procedure is to partition each of $I$, $J$, and $K$ into random
  subsets $I_1, \dots, I_{n^{\alpha}}$, $J_1, \dots, J_{n^{\alpha}}$,
  and $K_1, \dots, K_{n^{\alpha}}$, each of size $n^{1-\alpha}$, and
  for every triple of $q, r, s \in [n^{\alpha}]$, we run the algorithm
  of \cref{lem:find-uniquers-2} on $G[I_q, J_r, K_s]$.\\

  To verify this works, fix some particular hyperedge $(i,j,k)$ with
  $(i,j)$ unpopular, and let $I^* \in \{I_1, \dots, I_{n^{\alpha}}\},
  J^* \in \{J_1, \dots, J_{n^{\alpha}}\}, K^* \in \{K_1, \dots,
  K_{n^{\alpha}}\}$ be the subsets in the random partition that $i$,
  $j$, and $k$ are mapped to, respectively.  Since $(i,j)$ is
  unpopular, there are at most $10n^{\alpha}$ distinct values $k'\neq
  k$ such that $(i,j,k')$ is a hyperedge.  So, the probability that
  any such $k'$ is assigned to $K^*$ is at most $1 - \left(1 -
  \frac{1}{n^a}\right)^{10 n^\alpha} \leq 1 - e^{-20}$.  As long as
  this doesn't happen, $(i,j)$ will be involved in a unique
  hyperedge in $G[I^*, J^*, K^*]$, and so the algorithm will list
  $(i,j,k)$.\\

  Now, to calculate cost, observe that we run the algorithm of
  \cref{lem:find-uniquers-2} on $(n^{\alpha})^3$ instances, each on
  $n^{1-\alpha}$ vertices.  So, the overall runtime is
  $\widetilde{O}\paren{\left(n^{\alpha}\right)^3 \cdot
    \left(n^{1-\alpha}\right)^{3-\varepsilon}} =
  \widetilde{O}\paren{n^{3 - \eps + \eps\alpha}}$.
  
\end{proof}

With these subroutines, our specified-parts oracle algorithm is pretty
straightforward.

\begin{restatable}{theorem}{babytriangleub}\label{thm:baby-triangle-ub}
    Given cost-$n^c$ specified-parts induced all-pairs containment oracle access to an unknown $3$-partite, $3$-uniform hypergraph $H = (V, E)$, we can exactly recover $H$ with high probability at cost 
    $\softo{|V|^{3(1-\gamma)}|E|^{\gamma}},$
    where $\gamma = \frac{3-c}{4-c}$.
\end{restatable}

\begin{proof}[Proof of \cref{thm:baby-triangle-ub}.]

Let $G = (V = I\cup J\cup K, E)$, where $|E| = n^\tau$, be an unknown,
$3$-partite, $3$-uniform hypergraph, and let $\mathcal{A} =
\mathcal{A}_G$ be a cost-$n^c$ induced specified-parts all-$2$
containment oracle, where $c = 3-\eps$. First, we call the algorithm
given in \cref{lem:estimate} on $G$ to obtain an estimate $a_{i, j}$
for each $(i, j)\in I\times J$ of the number of hyperedges that $(i,
j)$ is in. Set $\alpha = (\tau+\eps-2)/(1+\eps)$, and call a pair of
vertices $(i, j)\in I\times J$ \defn{popular} if $a_{i, j}\ge
n^\alpha$, and \defn{unpopular} otherwise. The idea is to list the
hyperedges involving popular pairs and unpopular pairs
separately. First, for each popular pair $(i, j)$, we iterate through
all $k\in K$ and make a constant-cost call to $\mathcal{A}$ to check
whether $(i, j, k)$ is a hyperedge. Then we call the algorithm
described in \cref{lem:find-kinda-uniquers} on $G$ and $\alpha$ to
list all hyperedges involving an unpopular pair.\\

Since each popular pair is in at least $n^\alpha$ hyperedges and there
are $n^\tau$ hyperedges total, there are at most $n^{\tau-\alpha}$
popular pairs. So listing the hyperedges involving a popular pair
takes time $O(n^{\tau-\alpha+1})$. By \cref{lem:find-kinda-uniquers},
listing all the hyperedges involving an unpopular pair takes time
$\widetilde{O}(n^{3-\eps+\eps\alpha})$. (Additionally, estimating the
number of hyperedges that each pair is in is time
$\widetilde{O}(n^{3-\eps})$ by \cref{lem:estimate}, but since $\alpha
> 0$ this cost is dominated by the cost of
\cref{lem:find-kinda-uniquers}.) So the total runtime of the algorithm
is $\widetilde{O}(n^{\tau-\alpha+1}+n^{3-\eps+\eps\alpha})
= \widetilde{O}(n^{3-\eps(3-\tau)/(1+\eps)})$.
  
\end{proof}

Now we give our all-parts oracle algorithm. In this case, we will once again identify some collection of ``popular'' pairs, which are involved in many hyperedges. But now, we will do some additional processing to ensure that these popular pairs are not too highly clustered. That is, while the graph on popular pairs (which is an explicit graph we have access to from our estimates, as opposed to being queried via the oracle) has too many triangles, we will process all of the pairs involved in triangles of popular pairs with a single particular vertex, and then remove those popular pairs from consideration in all later queries. This introduces two downsides over \cref{thm:baby-triangle-ub}: one is that, although our query cost is low, we introduce substantial computational overhead from finding triangles in the graph of popular pairs. The other is that, since we must remove some popular pairs from consideration, this algorithm needs the general (i.e. not-necessarily-induced) version of all-pairs containment queries. The statement is as follows:

\allpartsub*

\begin{proof}[Proof of \cref{thm:allparts-ub}.]

Again, let $G = (V = I\cup J\cup K, E)$, where $|E| = t = n^\tau$, be
an unknown, $3$-partite, $3$-uniform hypergraph, and let $\mathcal{A}
= \mathcal{A}_G$ be a cost-$n^c$ induced specified-parts all-$2$
containment oracle, where $c = 3-\eps$. We first call the algorithm
given in \cref{lem:estimate} on $G$, this time for all pairs of parts,
to obtain an estimate $a_{u, v}$ for every $(u, v)\in V^2$. Set
parameters $\alpha = (x+\tau-3+\eps)/\eps$ if $2\le\tau\le 3$ and
$\alpha = (x-1+\eps)/\eps$ if $0\le\tau\le 2$, and $\beta = x+\tau-1$
if $2\le\tau\le 3$ and $\beta = x+1$ if $0\le\tau\le 2$, where $x =
-\eps+(-1+\sqrt{9+4\eps^2})/2$. Call a pair of vertices $(u, v)$
\defn{popular} if $a_{u, v}\ge n^{\alpha}$ and \defn{unpopular}
otherwise. We call the algorithm from \cref{lem:find-kinda-uniquers}
on $G$ and $\alpha$ to list all hyperedges that involve an unpopular
pair. So from this point on we can think of the graph as having only
popular pairs of vertices.\\

Note that now, since every pair of vertices is popular, there can be
at most $n^{\tau-\alpha}$ pairs of vertices in each of $I\times J$,
$J\times K$, and $I\times K$, since otherwise we would have $|E| >
n^\tau$. We iterate through all pairs of vertices and count for each
vertex the number of pairs that it is in. Call a vertex
\defn{fashionable} if it is in at least $1000n^{\tau-1-\alpha}$
popular pairs and \defn{unfashionable} otherwise. Note that $I$ can
contain at most $n/10$ fashionable vertices, since otherwise there
would be more than $n/10 \cdot 1000n^{\tau-1-\alpha} =
100n^{\tau-\alpha}$ popular pairs of vertices in either $I\times J$ or
$I\times K$. Let $I^*\subseteq I$ be the set of fashionable vertices
in $I$, and break $J$ into two halves $J_0$ and $J_1$ and $K$ into two
halves $K_0$ and $K_1$. We recurse on $G[I^*\cup J_{b}\cup K_{b'}]$
for every $b', b'\in\{0, 1\}$ to list all hyperedges involving an
fashionable vertex $i$. So now all that remains is to list the
hyperedges that contain an unfashionable vertex in $I$.\\

To do this, we construct an auxiliary tripartite graph $G' = (V', E')$
(an ordinary graph) on vertex parts $I' = I\backslash I^*$, $J' = J$,
and $K' = K$. The edges in $G'$ are exactly the popular pairs of
vertices in $G$ in $(I'\times J')\cup (I'\times K')\cup (J'\times
K')$. Note that every hyperedge not yet listed in $G$ is a triangle in
$G'$. The plan now is to find all hyperedges that contain some $i\in
I'$ that is in many triangles in $G'$; then, separately find the
hyperedges containing $i$ in not-so-many triangles.\\

Now we check whether there is an $i\in I'$ in at least $n^\beta$
triangles. If so, we partition $I'$ into parts $I'_1, \ldots,
I'_{n^{2+\alpha-\tau}}$, each of size $n^{\tau-1-\alpha}$. While there
is an $i\in I'$ in at least $n^\beta$ triangles, we recurse on
$G[I'_\ell\cup N(i)]$ for every $\ell\in [n^{2+\alpha-\tau}]$ (where
we use $N(i) = N_{G'}(i)$ for the neighborhood of $i$ in $G'$) to list
all the hyperedges in $G$ that involve pairs of vertices that are
edges in $G'[N(i)]$. Note that this includes all the hyperedges that
involve vertex $i$. Then we delete the edges in $G'[N(i)]$ from $G'$,
and then check whether there is still an $i\in I'$ in at least
$n^\beta$ triangles. After exiting this while-loop, to list the
remaining hyperedges, we iterate through all triangles $(i, j, k)$ in
$G'$ that contain an $i$ in at most $n^\beta$ triangles and make a
constant-cost call to $\mathcal{A}$ to check whether $(i, j, k)$ is a
hyperedge in $G$.\\

Then, the (worst-case) runtime of this algorithm, not considering
computational costs, is given by the recurrence
\[T(n, t)\le n^{2-\beta}T\left(\frac{t}{n^{1+\alpha}}, \frac{tn^\beta}{n^2}\right)+ 4T\left(\frac{n}{2}, \frac{t}{4}\right) + \widetilde{O}(n^{3-\eps+\eps\alpha}) + n^{1+\beta},
\]
where, going step by step through the algorithm, we can account for
the terms as follows:
\begin{itemize}

\item

estimating the number of hyperedges that each pair of vertices is
involved in and finding all hyperedges involving unpopular pairs takes
time $\widetilde{O}(n^{3-\eps+\eps\alpha})$;

\item

recursing to find all hyperedges that involve a fashionable vertex
takes time $\sum_{j = 1}^4 T(n/2, t_j)$, where $\sum_j t_j$ is at most
$t$; applying a convexity argument (Jensen's inequality), we get that
$\sum_{j = 1}^4 T(n/2, t_j)$ is maximized when the $t_j$'s are all the
same, i.e., $t_j = t/4$ for all $j$;

\item

the recursive calls made in the while-loop take time
$\sum_{\ell_1}^{K} n/(t/n^{1+\alpha})\cdot T(t/n^{1+\alpha},
t_{\ell_1, \ell_2})$, where $K$ is the number of iterations and for
each $\ell_1\in [K]$, $\ell_2\in [n/(t/n^{1+\alpha})]$, $t_{\ell_1,
  \ell_2}$ is the number of triangles in the subgraph of $G$ induced
by $N(i_{\ell_1})\cup I'_{\ell_2}$, where $i_{\ell_1}$ is the vertex
being considered in the $\ell_1$th iteration of the loop. We have that
$K\le t/n^{\alpha+\beta}$, since $G'$ initially contains at most
$t/n^\alpha$ edges between parts $J'$ and $K'$, and in each iteration we
remove $n^\beta$ of these edges (and once all these edges are removed
we have to have exited the loop). By applying a convexity argument, we
can upper bound the running time of this step by this sum where we
take $K$ to be the maximum possible and set all $t_{\ell_1, \ell_2}$
equal, which implies $t_{\ell_1, \ell_2} = tn^\beta/n^2$. This comes
out to be $t/n^{\alpha+\beta}\cdot n/(t/n^{1-\alpha})T(t/n^{1+\alpha},
tn^\beta/n^2) = n^{2-\beta}T(t/n^{1+\alpha}, tn^\beta/n^2)$; and

\item

finally, listing $n^{1+\beta}$ triangles in $G'$ at the end requires
time at least $n^{1+\beta}$.

\end{itemize}  

Now we inductively show that this recurrence satisfies $T(n, t)\le
K(tn^x+n^{2+x})$, where $x = -\eps+(-1+\sqrt{9+4\eps^2})/2$ and $K$ is
some sufficiently large constant. Assume the inductive hypothesis that
$T(n', t')\le K(t'n'^x+n'^{2+x})$ for all $t' < t$, $n' < n$. It
suffices to show that each of the four summands (where we ignore
polylogarithmic factors by taking $K$ large enough) are at most
$(K/4)(tn^x+n^{2+x})$. Applying the inductive hypothesis, we get that
$4T(n/2, t/4)\le K((t/4)(n/2)^x+(n/4)^{2+x})$ and
$n^{2-\beta}T(t/n^{1+\alpha}, tn^\beta/n^2)\le
K(n^{2-\beta}(tn^{\beta-2}(t/n^{1+\alpha})^x+(t/n^{1+\alpha})^{2+x}))$. It
is clear that the right-hand side of the first inequality is at most
$(K/4)(tn^x+n^{2+x})$ when $K$ is large enough. To see that the other
terms are also at most $(K/4)(tn^x+n^{2+x})$, it suffices to show that
our parameters satisfy the following systems of inequalities (which come from looking at the exponents in these inequalities).
\begin{itemize}

\item If $2\le\tau\le 3$, we have the following:

  \begin{align*}
  \beta+1 &\leq x+\tau \\
  3-\eps+\eps\alpha &\leq x+\tau\\
  2-\beta+\max(x(\tau-1-\alpha)+\tau-2+\beta, (2+x)(\tau-1-\alpha)) &\leq x+\tau.
  \end{align*}

\item Otherwise, if $0\le\tau\le 2$, we have the following:

  \begin{align*}
  \beta+1&\leq 2+x\\
  3-\eps+\eps\alpha &\leq 2+x\\
  2-\beta+(2+x)(\tau-1-\alpha)&\leq 2+x.  
  \end{align*}  

\end{itemize}

As chosen, $\alpha$ and $\beta$ satisfy the first two constraints in
both systems. So all that remains is to check that the last constraint
in each system is satisfied. In the first system, first consider the
case where $x(\tau-1-\alpha)+\tau-2+\beta \ge
(2+x)(\tau-1-\alpha)$. Then the constraint is $x(\tau-2-\alpha)\le
0$, which, substituting, yields $\tau(1-\eps) \ge 3-3\eps-x$. Since
$\tau\ge 2$, it suffices to show that $2(1-\eps) \ge 3-3\eps-x$, which
simplifies to $\eps \ge 1-x$, which holds. Next we consider the case
where $x(\tau-1-\alpha)+\tau-2+\beta \le (2+x)(\tau-1-\alpha)$. Then
the constraint becomes $2-\beta+(2+x)(\tau-1-\alpha)\le x+\tau$, which,
substituting, yields $(\tau-3)x-\alpha(2+x)+1\le 0$. With some
algebra, it can be seen that this holds for for $\tau\ge 2$, $\eps > 0$. Finally, the left-hand side of the last constraint in the second
system is at most $2-\beta+(2+x)(1-\alpha)\le 2+x$, taking $\tau = 2$,
so it suffices to show this is at most the right hand side. This
inequality yields $(x-1-+\eps)/\eps\ge (x-1)/(2+x)$, which holds since
we have $x-1+\eps\ge x-1$ and $\eps\le 2+x$.
\end{proof}  

Next we give an upper bound that considers computational costs. We first introduce a result that we will use on listing triangles, due to Bj\"{o}rklund, Pagh, Vassilevska Williams, and Zwick.

\begin{lemma}[Theorem 1, \cite{bjorklund2014listing}]\label{lem:list-triangles}

There exists a deterministic algorithm that lists all $t$ triangles in a graph of $n$ vertices in time $\widetilde{O}(n^\omega+n^{3(\omega-1)/(5-\omega)}t^{2(3-\omega)/(5-\omega)})$. Setting $\omega = 2$, the running time is $\widetilde{O}(n^2+nt^{2/3})$.
    
\end{lemma}

\allpartsubcomp*
\begin{proof}[Proof of \cref{thm:allparts-ub-comp}]
Now that we consider computational costs, the running time of our algorithm is given by the recurrence
\[T(n, t)\le n^{2-\beta}T\left(\frac{t}{n^{1+\alpha}}, \frac{tn^\beta}{n^2}\right)+ 4T\left(\frac{n}{2}, \frac{t}{4}\right) + \widetilde{O}(n^{3-\eps+\eps\alpha}) + \widetilde{O}(n^2+n^{5/3+2\beta/3}) + t^2n^{2\alpha-\beta},
\]
where we set parameters as follows:
\begin{itemize}
    \item $x = \frac{-2-5\eps + \sqrt{36+12\eps+25\eps^2}}{4}$,
    \item if $2\le\tau\le 3$, $\alpha = (x+\tau-3+\eps)/\eps$, and if $0\le\tau\le 2$, $\alpha = (x-1+\eps)/\eps$, and
    \item if $2\le\tau\le 3$, $\beta = (3x+3\tau-5)/2$, and if $0\le\tau\le 2$, $\beta = (3x+1)/2$.
\end{itemize}
This is the same recurrence as we had before, when only considering information-theoretic costs, except that the term $n^{1+\beta}$ from before is replaced with $\widetilde{O}(n^2+n^{5/3+2\beta/3}) + t^2n^{2\alpha-\beta}$ (and we change the parameters). The $t^2n^{2\alpha-\beta}$ term is the cost of finding, across every iteration of the while-loop, the number of triangles that every vertex $i\in I'$ is in, using matrix multiplication, assuming $\omega = 2$. The $\widetilde{O}(n^2+n^{5/3+2\beta/3})$ term is the cost of listing all the remaining triangles at the end of the algorithm, applying \cref{lem:list-triangles} with $\omega = 2$.

As before, we guess that this recurrence satisfies $T(n, t)\le K(tn^x+n^{2+x})$ for some sufficiently large constant $K$ and we show by induction that this is the case. Assume that $T(n', t')\le K(t'{n'}^x + {n'}^{2+x})$ for all $t' < t$, $n' < n$. It suffices to show that each of the five summands (where we ignore polylog factors by taking $K$ large enough) are at most $(K/5)(tn^x+n^{2+x})$.  Applying the inductive hypothesis, we get that
$4T(n/2, t/4)\le K((t/4)(n/2)^x+(n/4)^{2+x})$ and
$n^{2-\beta}T(t/n^{1+\alpha}, tn^\beta/n^2)\le
K(n^{2-\beta}(tn^{\beta-2}(t/n^{1+\alpha})^x+(t/n^{1+\alpha})^{2+x}))$. It
is clear that the right-hand side of the first inequality is at most
$(K/5)(tn^x+n^{2+x})$ when $K$ is large enough. To see that the other
terms are also at most $(K/5)(tn^x+n^{2+x})$, it suffices to show that
our parameters satisfy the following systems of inequalities (which come from looking at the exponents in these inequalities).
\begin{itemize}

\item If $2\le\tau\le 3$, we have the following:

  \begin{align*}
  5/3+2\beta/3 &\leq x+\tau \\
  3-\eps+\eps\alpha &\leq x+\tau\\
  2\tau+2\alpha-\beta &\leq x+\tau \\
  2-\beta+\max(x(\tau-1-\alpha)+\tau-2+\beta, (2+x)(\tau-1-\alpha)) &\leq x+\tau.
  \end{align*}

\item Otherwise, if $0\le\tau\le 2$, we have the following:

  \begin{align*}
  5/3+2\beta/3 &\leq 2+x \\
  3-\eps+\eps\alpha &\leq 2+x\\
  2\tau+2\alpha-\beta&\leq 2+x\\
  2-\beta+(2+x)(\tau-1-\alpha)&\leq 2+x.  
  \end{align*}  

\end{itemize}

In each of these systems, the parameters $\alpha$ and $\beta$ are chosen to satisfy the first two constraints, and by plugging in expressions for $\alpha$ and $\beta$ in the third constraint and doing a little algebra, we can readily check that the third constraint is also satisfied. So what remains is to check that the last constraint is satisfied. 

First we consider the first system. We have two cases:
\begin{enumerate}
    \item $x(\tau-1-\alpha)+\tau-2+\beta \ge (2+x)(\tau-1-\alpha)$

    In this case, the constraint we need to satisfy is $2-\beta+x(\tau-1-\alpha)+\tau-2+\beta\le x+\tau$. This is the same as $(\tau-2-\alpha)x\le 0$, which yields $\tau-2-\alpha\le 0$ since we have $x > 0$. Substituting in our expression for $\alpha$, we get that this is the same as $x\ge 1-\eps$, which holds.

    \item $x(\tau-1-\alpha)+\tau-2+\beta \le (2+x)(\tau-1-\alpha)$

    In this case, the constraint we need to satisfy is $2-\beta+(2+x)(\tau-1-\alpha)\le x+\tau$. Simplifying, this yields $x(\tau-2-\alpha)\le (2\alpha-\tau+\beta)$. Assuming that $x > 1-\eps$, plugging in $\alpha$ and $\beta$, and solving for $x$, we get  that $x\ge\frac{-2-5\eps + \sqrt{36+12\eps+25\eps^2}}{4}$, so this constraint holds for our choice of $x$.
\end{enumerate}
Next we consider the second system. Plugging in our expressions for $\alpha$ and $\beta$ and doing some algebra, we get that this inequality is satisfied for all positive $x$, which means that this constraint holds.
\end{proof}

\subsection{Interpretation in terms of black-box fine-grained reductions}

This exact recovery problem is relevant in the fine-grained setting because it allows us to solve triangle \emph{listing}. This result partially answers the question from \cref{sec:finegrained} of finding a dense analogue of Duraj, Kleiner, Polak, and Vassilevska Williams's theorem from \cite{duraj2020equivalences}. That is, \cref{thm:allparts-ub-comp} implies the following:

\begin{corollary}
    Assuming $\omega = 2$, for any $2 \leq c \leq 3$ and any $0 \leq k \leq 3$, the problem of listing $n^k$ $R$-triangles $\paren{n^c, n^{\paren{\frac{5c-9 + \sqrt{297-162c+25c^2}}{4}}+ \max(k - 2, 0)}}$-black-box reduces to all-edge $R$-triangle detection.
\end{corollary}

If the original graph contained exactly $n^k$ triangles, this would be immediate. (Using the fact that, while not operating in the ``induced'' setting, \cref{thm:allparts-ub-comp} only ever restricts the set of hyperedges considered in a query by excluding all hyperedges involving a given pair of vertices. This can be accomplished in the $R$-triangle setting simply by deleting the relevant edge. Note also that \cref{thm:allparts-ub-comp} is computationally efficient, and hence can be used in a black-box $R$-triangle reduction --- this requirement is why we cannot use the stronger bounds of \cref{thm:allparts-ub}. )
In the case that the original graph contains more than $n^k$ triangles, this corollary can be obtained by randomly sub-sampling the vertices of the graph at each dyadic rate and running the exact recovery algorithm on each (stopping if we run out of time). 
For some rate, we will with high probability isolate more than $n^k$ triangles but fewer than $n^k \polylog(n)$.
In that case, our exact recovery algorithm runs in the desired time and finds at least $n^k$ triangles.\\

Our specified-parts upper bound is also relevant, as an answer to the discussion about ``rotated'' matrix products. From \cref{thm:baby-triangle-ub}, we get the following:
\begin{corollary}
    For any $2 \leq c \leq 3$ and any $0 \leq k \leq 3$, there is an $\paren{n^c, n^{\paren{3 - \frac{(3 -3k)(3-c)}{4-c}}}}$-black-box reduction from the problem of listing $n^k$ $R$-triangles to all-$IJ$ $R$-triangle detection (that is, determining which edges between vertex parts $I$ and $J$ belong to $R$-triangles).
\end{corollary}

Plugging in $k=2$, and using \cref{thm:alledge-to-listing} (from \cite{duraj2020equivalences}), we therefore in particular get the following:

\begin{corollary}
    For any $2 \leq c \leq 3$, all-edge $R$-triangle detection $\paren{n^c, n^{\paren{\frac{9-2c}{4-c}}}}$-black-box reduces to all-$IJ$ detection. 
\end{corollary}

Hence, if we can view a pair of matrix products as rotated versions of the same $R$-triangle problem, then a time-$n^c$ algorithm for one generically implies a time-$n^{\paren{\frac{9-2c}{4-c}}}$ algorithm for the other.
In the case of boolean-($\min, +$)-product and dominance product, we know both are solvable in time $O(n^{2.69})$. Since there is a reduction from boolean-$(\min, +)$-product to a rotated version of dominance product, the above result shows that a time-$O(n^{2.53})$ dominance product algorithm would give a time-$O(n^{2.685})$ boolean-$(\min, +)$-product algorithm.

\section{Upper bounds for exact recovery in general}\label{sec:specified-parts-upperbound}
In this section, we generalize \cref{thm:baby-triangle-ub} to hypergraphs of higher uniformity.
Our main result is as follows:

\specifiedpartsub*

In \cref{sec:lbs}, we show that this upper bound is tight (up to polylog factors) for any $s \leq c \leq r$ and any $|V| \leq |E| \leq |V|^r$.
The proof of \cref{thm:specifiedpartsub} is similar to that of \cref{thm:baby-triangle-ub}, but there are several complications that arise in these other uniformities.
For one, it is no longer as immediate that we can use our all-$s$ containment oracle to approximate all-$s$ codegrees (although this is in fact possible).
Another somewhat more serious issue is that it is no longer the case that we can expect to ``isolate'' hyperedges for unpopular $s$-tuples by randomly partitioning all of the parts. 
When we had $r - s = 1$, for a fixed $s$-tuple of vertices we knew that the set of hyperedges containing those vertices simply corresponded to a subset of the vertices in the remaining part.
But now fixing an $s$-tuple of vertices can induce some higher-order structures which can result in many hyperedges falling into the same collection of buckets.
In order to deal with this, we argue that the only cases in which such collisions occur can be attributed to collections of vertices with badly-behaved codegrees, and that we can efficiently identify all such collections and handle them by brute force.\\

We start with the following important subroutine. 

\begin{lemma}\label{lem:hot-potato}
  For $k$-uniform, $k$-partite, $n$-vertex hypergraphs, for any positive integer $h \leq n$,
  there exists an algorithm in the independent set query model that with high probability returns a family $\mathcal{F}$ of vertex subsets such that
  \begin{itemize}
    \item every element $S \in \mathcal{F}$ has codegree at least $h^{k-|S|}\log^{-10k}(n)$, and
    \item every hyperedge is a superset of some element $S \in \mathcal{F}$. 
  \end{itemize}
  The algorithm makes $\softo{h^k}$ non-adaptive independent set queries of size at most $n/h$. 
\end{lemma}
\begin{proof}
  Partition each vertex part randomly into $h$ chunks of size $n/h$.
  For every tuple of $k$ of those chunks, we would like to determine the identity of some hyperedge contained in that $k$-tuple of chunks (if any such hyperedge exists).
  If the $k$-tuple contains a unique hyperedge, this is straightforward: we can simply binary search (that is, for all $i \in [k]$, $j \in [\log(n)]$, $b \in \{0,1\}$, discard all vertices in the $i$th vertex part with a $b$ in the $j$th binary digit of their name, and make an independent set query on the rest of the vertices in the chunks --- this procedure is non-adaptive, and if there is a unique hyperedge one can read off the identities of all vertices in the hyperedge from the result).
  In order to deal with $k$-tuples of chunks that contain multiple hyperedges, we can run the same procedure, but first subsample the vertices in each of the parts.
  That is, for every $r_1, \dots, r_k \in [\log(n)]$, for $\log^{10}(n)$ iterations, we choose a random sample of $2^{-r_i}$ vertices from the $i$th chunk in the $k$-tuple for all $i$, and we run the above procedure to find a unique hyperedge if one exists.
  Observe that there is some setting of all the $r_i$ such that a unique hyperedge will survive with constant probability --- so, with high probability, for all $k$-tuples containing at least one hyperedge we will find one on some iteration of this procedure.\\

  Now that we've found this list of representative hyperedges, for every subset of vertices in the hypergraph, we'll record their codegree \emph{among these representative hyperedges}.
  If we find that some subset $S$ has codegree at least $h^{k-|S|}\log^{-10k}(n)$ among these identified edges, we know that it has codegree at least that high in the full hypergraph, and so can safely add it to $\mathcal{F}$.
  We would now like to show that, for any hyperedge $E$ of the hypergraph, with probability at least $1/\polylog(n)$ this process will result in some subset of $E$ being added to $\mathcal{F}$.
  Note that this is sufficient for the overall lemma: if we then repeat the entire process some large $\polylog(n)$ many times, we can ensure that the probability of any given hyperedge failing to have a subset added to $\mathcal{F}$ is much less than $n^{-k}$, meaning that we can union bound over all of these events.\\

  Say that a pair of hyperedges $E$ and $E'$ \defn{collide} if they both belong to the same $k$-tuple of chunks.
  If the probability of $E$ colliding with another edge is less than $1/2$, then we are already done: with at least $1/2$ probability, $E$ will be the unique hyperedge in its $k$-tuple of chunks, meaning that it will be necessarily chosen as the representative.
  Any hyperedge found as one of the representatives has codegree $1 \gg h^{k - k}\log^{-10k}(n)$ among the representatives, so will be itself added to $\mathcal{F}$.
  Otherwise, say that a subset $U \subseteq E$ is \defn{hot} if, with probability at least $\log^{10|U|-10k}(n)$, there will be a collision between $E$ and some other edge $E'$ with $E \cap E' = U$.
  Observe that by union bound at least one subset of $E$ must be hot.
  Fix $U$ to be some hot subset of $E$ of minimal size. We will show that, with probability at least $1/\polylog(n)$, this process adds $U$ to $\mathcal{F}$.\\

  First observe that, for an edge $E'$ overlapping with $E$ exactly on $U$, the $k-|U|$ chunks the remaining vertices of $E'$ map to are independent of those that the remaining vertices of $E$ map to.
  So, since the probability that there exists some such $E'$ colliding with $E$ is at least $\log^{10|U|-10k}(n)$, we can also observe that with high probability the number of $k$-tuples of chunks that contain \emph{some} edge containing $U$ will be at least $O(h^{k-|U|}\log^{10|U|-10k}(n))$.
  Call these $k$-tuples of chunks the \defn{useful} $k$-tuples.\\

  Fix any proper subset $W \subseteq U$, and consider the expected number of useful $k$-tuples that contain some hyperedge $E'$ such that $E' \cap U = W$.
  The probability that $E'$ shares chunks with all of the vertices in $U$ is $h^{|W| - |U|}$.
  So, if this expectation was greater than $\paren{h^{|W| - |U|}}\paren{h^{k-|W|}\log^{10|U|-10k}(n)} = h^{k-|U|}\log^{10|W|-10k}(n)$, then we would have that $W$ is hot, contradicting the assumption that $U$ has the minimum size among hot subsets of $E$.
  Thus, the expected number of useful $k$-tuples that contain \emph{any} hyperedge $E'$ with $E' \cap U \neq U$ is at most
  \[\sum_{\substack{W \subseteq U \\ W \neq U}} h^{k-|U|}\log^{10|W|-10k}(n) \leq 2^{|U|} h^{k-|U|} \log^{10|U| - 10k - 10}(n) \ll h^{k-|U|}\log^{10|U|-10k}(n).\]
  
  So, with high probability, there are at least $O(h^{k-|U|}\log^{10|U|-10k}(n)) \gg h^{k-|U|}\log^{-10k}(n)$ distinct $k$-tuples of chunks that contain some hyperedge containing $U$, and no hyperedge not containing $U$.
  For each of those chunks, our process will extract a representative hyperedge containing $U$.
  So, $U$ will be added to $\mathcal{F}$.
\end{proof}

We are now ready to show the main result.

\begin{proof}[Proof of \cref{thm:specifiedpartsub}]
  We will set $h = |V|^{\frac{r-s-c}{2r-s-c}}|E|^{\frac{1}{2r-s-c}}$, and partition each of the specified parts arbitrarily into $h$ chunks.
  For each $s$-tuple of those chunks, we will do the following.
  Observe that, if we look at the output of our all-$s$ containment oracle for one specific $s$-tuple of vertices, it is exactly an independent set oracle on the $(r-s)$-uniform hypergraph induced by the restriction.
  So, setting $k=r-s$, we can run the algorithm from \cref{lem:hot-potato}.
  That is, for every $s$-tuple of chunks among the specified parts, we run the algorithm from \cref{lem:hot-potato} on the remaining $r-s$ parts of the graph.
  Every time the algorithm would make an independent set query on a vertex subset of those $r-s$ parts, we instead make a query to the all-$s$ containment oracle, including that subset along with the chosen chunks from the remaining $s$ parts.
  Since the queries are made non-adaptively, this process is well-defined, and since all queries are made to vertex subsets of size at most $O(|V|/h)$, the total cost is $h^s \cdot \softo{h^k} \cdot O(|V|/h)^c = \softo{h^{r-c}|V|^c} = \softo{|V|^{r(1-\gamma)}|E|^\gamma}$.\\

  Now, for every $s$-tuple $X$ of vertices among the specified parts, we have determined a family $\mathcal{F}_X$ such that every hyperedge containing $X$ also contains some element of $\mathcal{F}_X$, and every element $S \in \mathcal{F}_X$ satisfies that the codegree of $S \cup X$ is at least $\widetilde{\Omega}(h^{r-s - |S|})$.
  For every $s$-tuple $X$ among the first $s$ parts, for every $S \in \mathcal{F}_X$, our algorithm will now brute force over all $|V|^{r - s - |S|}$ many $(r-s-|S|)$-tuples $Y$ of vertices in the remaining parts, and make a constant-size query to the oracle to check whether $X \cup S \cup Y$ is a hyperedge.
  Observe that, since every hyperedge containing $X$ also contains some $S \in \mathcal{F}_X$, this procedure will find all of the hyperedges of the hypergraph.
  To determine the cost of this procedure, note that every time we process an $S \in \mathcal{F}_X$, we expend cost $O(|V|^{r-s-|S|})$, and we discover at least $\widetilde{\Omega}(h^{r-s-|S|})$ new hyperedges we had not previously found.
  So, the total cost of this step is at most 
  \[\max_{0 \leq \ell \leq r-s} \paren{\frac{|E|}{\widetilde{\Omega}(h^{r-s-\ell})}\cdot O\paren{|V|^{r-s-\ell}}} = \softo{|E| \cdot \paren{\frac{|V|}{h}}^{r-s}} = \softo{|V|^{r(1-\gamma)}|E|^\gamma}.\] 
\end{proof}

\section{Lower bounds for exact recovery}\label{sec:lbs}

In \cref{sec:specified-parts-upperbound}, we gave an algorithm for solving exact recovery of sparse hypergraphs using induced all-$s$ containment queries.
Here, we demonstrate that upper bound to be tight even with specified-parts all-$s$ \emph{codegree} queries.
Our proof will make use of the following simple combinatorial fact:

\begin{lemma}\label{lem:hypergraph-nodense}
    Let $K$ be the random $s$-partite, $s$-uniform hypergraph with $n$ vertices per part, obtained by including each hyperedge with independent probability $n^{-\beta}$.
    Then, with high probability there does not exist a subhypergraph on $x$ vertices and more than $10x^sn^{-\beta}\log n + 10x \log n$ hyperedges for any $x$.
\end{lemma}
\begin{proof}
    For a given set of $x$ vertices, the number of induced hyperedges is distributed as $\text{Bin}(x^s, n^{-\beta})$.
    By a standard Chernoff bound, we have 
    \[\Pr[\text{Bin}(x^s, n^{-\beta}) \geq 10x^sn^{-\beta}\log n + 10x \log(n)] \leq 2^{-3x^sn^{-\beta}\log(n) - 3x \log(n)} \leq n^{-3x - 2}.\]
    So, union bounding over all choices of $x$ and all $n^x$ sets of vertices, the probability of any such dense subhypergraph existing is at most $\sum_{x=1}^n n^x \cdot n^{-3x - 2} < o(1)$. 
\end{proof}

\specifiedpartslb*

\begin{proof}
    Once again, we will use Yao's lemma.
    We would like to design a distribution over $r|V|$-vertex, $|E|$-edge, $r$-partite, $r$-uniform hypergraphs, such that any deterministic algorithm is unlikely to successfully recover the hypergraph from specified-parts all-$s$ codegree queries.
    To construct such a distribution, we will start by declaring each $s$-tuple of vertices spanning the first $s$ parts of the hypergraph as ``admissible'' with independent probability $p$, for $p = \paren{|V|^{-r}|E|}^{\max\paren{\frac{r-c}{2r-s-c}, \frac{s-1}{r-1}}}$.
    This choice is fixed and revealed to the algorithm, as opposed to being a part of the distribution --- the only property we require of those $s$-tuples is that guaranteed by \cref{lem:hypergraph-nodense} (so in particular, any explicit construction of such a quasirandom $s$-uniform hypergraph would also suffice to define the admissible $s$-tuples).
    Now, our distribution on hypergraphs will simply be to let each of the spanning $r$-tuples of vertices whose restriction to the first $s$ parts is an admissible $s$-tuple be a hyperedge with probability $|E||V|^{-r}p^{-1}$ independently\footnote{Note that this defines a distribution with $|E|$ edges in \emph{expectation} as opposed to \emph{exactly} $|E|$. For those concerned about getting exactly the right value of $|E|$, one modification could be to apply this construction using half of the vertices and $1/2^r$ of the edges, and then place the remaining edges arbitrarily among the other vertices. It is vanishingly unlikely that more than $|E|$ edges will be needed if we're doing this construction at half scale, and the observation of the other half of the graph will provide at most $\log(|V|)$ bits about the distribution (namely, the number of edges chosen).}.\\

    The analysis at this point will proceed similarly to the proof of \cref{lem:cant-use-is}.
    We observe that the distribution we've defined over hypergraphs has entropy $|V|^{r}pH(|E||V|^{-r}p^{-1})$.
    We set a threshold of $t = 2|E||V|^{-r}p^{-1}$, and agree to reveal whether or not any given $r$-tuple is a hyperedge once its probability conditioned on our queries thus far of being so exceeds $t$.
    Letting $X_i$ be the results of the $i$th query and $Y_i$ be the additional information revealed about threshold-exceeding $r$-tuples after the $i$th query, as in \cref{lem:cant-use-is} correctness of the algorithm implies that we have
    \[\sum_i H(X_i | X_1, \dots, X_{i-1}, Y_1, \dots, Y_{i-1}) \geq \Omega(|E|).\]
    Now, for a sequence of queries $(Q_1, \Rreplace_1), \dots, (Q_q, \Rreplace_q)$, and the resulting responses $X_i = x_i$ and $Y_i = y_i$, we upper bound 
    \[\sum_i H(X_i | X_1=x_1, \dots, X_{i-1}=x_{i-1}, Y_1=y_1, \dots, Y_{i-1}=y_{i-1})\]
    under the assumption that 
    \[\sum_i|Q_i|^c < \widetilde{o}\paren{|V|^{r(1-\gamma)}|E|^\gamma}.\]
    Consider a query of size $|Q_i|$. 
    Since each query is made to a subhypergraph of the original hypergraph, by \cref{lem:hypergraph-nodense} we know that $Q_i$ can involve at most $10|Q_i|^sp\log(|V|) + 10|Q_i|\log(|V|)$ admissible $s$-tuples.
    We already know that no non-admissible $s$-tuple is involved in a hyperedge, so the query will reveal information only for these tuples.
    For each admissible $s$-tuple, at most $\log(|E|) = O(\log(|V|))$ bits of entropy are revealed, as the support of the random variable is at most size $|E|$.
    We also know that each admissible $s$-tuple is involved in at most $|Q_i|^{r-s}$ distinct $r$-tuples in $Q$, and that for each of these that hasn't yet been determined the probability of being a hyperedge is at most $t$.
    So, by subadditivity of entropy we can also upper bound the conditional entropy of any given $s$-tuple's response by $|Q_i|^{r-s}t = 2|Q_i|^{r-s}|E||V|^{ - r}p^{-1}$.
    This overall lets us write 
    \begin{align*}
     \sum_i H(X_i | X_1=x_1, \dots, X_{i-1}=x_{i-1}, Y_1=y_1, \dots, Y_{i-1}=y_{i-1}) &\leq \\
     \sum_i 10 \log(|V|)\cdot \paren{|Q_i|^sp + |Q_i|} \cdot \min\paren{\log(|V|), |Q_i|^{r-s}|E||V|^{ - r}p^{-1}} &\leq \\
     \sum_{\substack{i \\ |Q_i| > \paren{|V|^{r}|E|^{-1}p}^{\paren{\frac{1}{r-s}}}}}\softo{|Q_i|^sp + |Q_i|} + \sum_{\substack{i \\ |Q_i| < \paren{|V|^{r}|E|^{-1}p}^{\paren{\frac{1}{r-s}}}}}\softo{(|Q_i|^sp + |Q_i|)\cdot |Q_i|^{r-s}|E||V|^{ - r}p^{-1}} &\leq\\
     \left(\sum_i |Q_i|^c\right) \cdot \paren{\frac{1}{\paren{|V|^r|E|^{-1}p}^{\paren{\frac{c}{r-s}}}} \cdot \softo{\paren{|V|^{r}|E|^{-1}p}^{\paren{\frac{s}{r-s}}}p + \paren{|V|^{r}|E|^{-1}p}^{\paren{\frac{1}{r-s}}}} + |E||V|^{-r}p^{-1}} &=\\
     \left(\sum_i |Q_i|^c\right) \cdot \softo{\paren{|V|^{r}|E|^{-1}p}^{\paren{\frac{s-c}{r-s}}}p + \paren{|V|^{r}|E|^{-1}p}^{-\max\paren{1, \frac{c-1}{r-s}}}} &\leq\\
     \widetilde{o}\paren{|V|^{r(1-\gamma)}|E|^\gamma} \cdot 
     \softo{\paren{|V|^{r}|E|^{-1}p}^{\paren{\frac{s-c}{r-s}}}p + \paren{|V|^{r}|E|^{-1}p}^{-\max\paren{1, \frac{c-1}{r-s}}}} &\leq\\
    \widetilde{o}(|E|).
    \end{align*}
    Thus, runs where the query cost is less than this bound cannot generate enough entropy to achieve our $\Omega(|E|)$ lower bound on $\sum_iH(X_i | X_1, \dots, X_{i-1}, Y_1, \dots, Y_{i-1})$, meaning that the algorithm must sometimes require larger query cost.
\end{proof}

\allpartslb*
\begin{proof}
    Now, we'll take $p = (|V|^{-r}|E|)^{\min\paren{\frac{r-c}{(r-s)\binom{r}{s} + (r-c)}, \frac{s-1}{(r-s) + \binom{r}{s}(s-1)}}}$, and we'll declare each $s$-tuple of vertices in the entire graph to be admissible with independent probability $p$.
    Fixing the admissible $s$-tuples, we will say an $r$-tuple of vertices is admissible if every $s$ of them are admissible.
    Our distribution over hypergraphs will be to let each admisible spanning $r$-tuple of vertices be a hyperedge with independent probability $|E||V|^{-r}p^{-\binom{r}{s}}$.\\

    The analysis now follows exactly as in the proof of \cref{thm:specifiedparts-lb}.
    We bound the conditional entropy of a query's response given the responses of all previous queries by the sum of the conditional entropies of the response for each admissible $s$-tuple in the query.
    In a query $Q$, each admissible $s$-tuple yields at most $\softo{\max\paren{1, |Q|^{r-s}|E||V|^{-r}p^{-\binom{r}{s}}}}$ bits of entropy, and our query can contain at most $\softo{|Q|^sp + |Q|}$ of them by \cref{lem:hypergraph-nodense}.
    Appealing to convexity as in the previous case, one finds that in order for all queries to yield sufficient entropy we must sometimes have $\sum_i |Q_i|^c \geq \widetilde{\Omega}\paren{|V|^{r(1-\gamma)}|E|^\gamma}$.
\end{proof}

\subsection{Interpretation in terms of black-box fine-grained reductions}

In the case $r=3$, $s=2$, we can think of these results as giving lower bounds on black-box reductions from $R$-triangle listing to all-edge $R$-triangle detection.
By default, however, they only apply to subgraph reductions --- i.e. black-box reductions where every query is made to a subgraph of the input graph.
Unlike other results earlier in this paper, in this case that condition appears nontrivial to remove: it's a priori plausible that a reduction might do better by constructing totally different graphs and querying those. 
In \cref{sec:duplicitousness}, we show how to replace this subgraph condition with a more mild restriction.
But for now, let us simply state what the above results would imply about subgraph reductions.

\begin{corollary}[to \cref{thm:allparts-lb}]
    For any $2 \leq c \leq 3$ and any $\eps >0$, there is no $\paren{n^c, n^{\paren{3 - \frac{3(3-c)(3-k)}{6 - c} - \eps}}}$-black-box subgraph reduction from listing $n^k$ $R$-triangles to all-edge $R$-triangle detection.
\end{corollary}

We can also show a similar result for reductions to all-$IJ$ (namely, the problem of determining which edges between $I$ and $J$ are involved in $R$-triangles). In this case, we the lower bound is tight with our known upper bond:

\begin{corollary}[to \cref{thm:specifiedparts-lb}]
    For any $2 \leq c \leq 3$ and any $\eps >0$, there is no $\paren{n^c, n^{\paren{3 - \frac{(3-k)(3-c)}{4 - c} - \eps}}}$-black-box subgraph reduction from listing $n^k$ $R$-triangles to all-$IJ$ $R$-triangle detection.
\end{corollary}

For the problem of rotated matrix products, we wanted instead to reduce the problem of all-$JK$ detection to all-$IJ$ detection. Observe that a proof analogous to \cref{thm:specifiedparts-lb} would also prove the same lower bound against this case: setting $|E| = n^2$, the random variable representing which subset of $JK$-edges are involved in hyperedges has entropy $\Omega(n^2)$ in our choice of distribution. So, if we can bound the entropy of the query responses by $o(n^2)$, this shows that the algorithm not only can't solve exact recovery, it can't solve all-$JK$ detection.

\section{Listing to all-edge detection with non-duplicitous reductions}\label{sec:duplicitousness}

The previous section showed lower bounds on black-box reductions from listing to all-edge detection (or all-$IJ$ detection), in the special case where \emph{all queries are made to subgraphs}.
In this section, we strengthen those results to hold for any reduction with the property that no edge weight occurs with higher multiplicity in a call of the reduction than in the input graph (we call such reductions \emph{non-duplicitous}). A strengthening to arbitrary black-box reductions may be possible, but we do not know how to prove it.\\

The key tool we use to obtain that strengthening is the following combinatorial lemma:
\begin{definition}
    Say a graph $G$ is $(a,b,c)$-\defn{squishable} if, for some collection of subgraphs $G_1, \dots, G_a \subseteq G$ which together contain at least $b$ distinct edges from $G$, there exists a $c$-vertex graph $H$ which contains each $G_i$ as an \emph{induced} subgraph.
\end{definition}
\begin{lemma}\label{lem:babypacking}
    The Erd\H{o}s--R\'enyi graph $G \sim G(n,p)$ is w.h.p. not $(qp, 10q^2p\log(n) + 10q \log(n), q)$-squishable for any $q$.
\end{lemma}

To understand what this lemma is saying, it is helpful to consider the case $p = n^{-1/2}$, $q=n$. Here, we're asking whether $G(n, n^{-1/2})$ is $(\sqrt{n}, 20n^{1.5}\log(n), \sqrt{n})$-squishable.
There are two reasonable ways to show that $G(n, n^{-1/2})$ is $(\sqrt{n}, n^{1.5}/\polylog(n), \sqrt{n})$-squishable.
For one, we could let $H = G_1 = G$, and all other $G_i$ be empty --- this with high probability has $n$ vertices and at least $n^{1.5}/\polylog(n)$ edges because $G$ itself does.
For another, we could find $\Omega(n^{.5}/\polylog(n))$ many linear-sized edge-disjoint matchings in $G$ (which is possible for any graph of $G$'s min-degree~\cite{tutte1952factors}), and then pack those into $n$ vertices in an induced fashion using known constructions of so-called Ruzsa--Szemer\'edi graphs~\cite{monotonicity-testing}.
\cref{lem:babypacking} says that no approach can do better than these two. Matchings can be packed together in an induced way very tightly, but for a random graph there is no denser substructure with that property.\\

We now give a lower bound proof assuming \cref{lem:babypacking}. (Note that we only deal with the case where we're listing exactly $n^2$ edges, as opposed to an arbitrary $k$.)

\begin{theorem}
    For any $2 \leq c \leq 3$ and any $\eps >0$, there is no non-duplicitous $\paren{n^c, n^{\paren{\frac{9-2c}{4 - c} - \eps}}}$-black-box reduction from listing $n^2$ $R$-triangles to all-$IJ$ $R$-triangle detection.
\end{theorem}
\begin{proof}
    The structure of the proof is the same as that of \cref{thm:specifiedparts-lb}; we will primarily discuss the differences.
    As before, for an appropriately chosen parameter $p$ (the same as in \cref{thm:specifiedparts-lb} --- note that we will always have $p>n^{-1/2}$ no matter $c$), we choose a random $pn^2$ many ``admissible'' pairs of vertices between $I$ and $J$.
    We place edges between those admissible pairs, and between all $IK$ and $JK$ pairs. 
    All edges are given distinct edge weights.\\

    Then, we fix a distribution over relations $R$ by selecting a random $n^k$ triangles in this graph, and letting the relation be satisfied by those triples of edge weights. For every other triangle in the graph, the corresponding triple of edge weights does not belong to $R$. However, for any triple of edge weights that does \emph{not} appear in a triangle in the graph, we \emph{do} declare that triple to belong to $R$. This will be crucial to the proof.\\

    Recall that our lower bound approach in \cref{thm:specifiedparts-lb} worked by bounding the total entropy generated by the queries.
    To do, we upper bounded the number of candidate triangles a query can involve.
    The key fact we used was that any query of size $q$ could only involve at most $10q^2p \log(n) + 10 q \log(n)$ admissible pairs of $IJ$ vertices, since the queries were made to subgraphs.
    Now, however, this is no longer true --- we will need to use \cref{lem:babypacking} to get an analogous result.\\
    
    For a query $Q = Q_I \sqcup Q_J \sqcup Q_K$ of size $q$, say an edge $e$ between $Q_I$ and $Q_J$ is ``good'' if every triangle in $Q$ involving $e$ has edge weights corresponding to some triangle in the original graph.
    Observe that any edge which is not ``good'' is guaranteed to output $1$ (i.e., to be involved in an $R$-triangle in the query), since our relation includes every triple of edge weights that does not appear in a triangle in the original graph.
    So, for the purposes of upper bounding entropy, we can safely ignore all non-good edges and only count the good ones.\\

    Now, observe that by convexity the ``worst case'' (i.e. most entropy generated), for a query $Q = Q_I \sqcup Q_J \sqcup Q_K$ of size $q$ containing a given number of candidate triangles, is where every good edge in the query is involved in exactly $pn$ candidate triangles.
    In this case, we upper bound the entropy from this query by $1$ bit for every such good edge.
    So, we can upper bound the total entropy of the query by $\sum_{v \in Q_K'} \Big( \text{\# of good edges in }N(v)\Big)$ for some $Q_K' \subseteq Q_K$ with $|Q_K'| = q/pn$.\\

    By \cref{lem:babypacking}, in order to show that there are at most $10q^2p \log(n) + 10 q \log(n)$ good edges in the query, it suffices to write the good edges as a union of $q/pn \leq qp$ many induced subgraphs, each of which corresponds to a subgraph of the original admissible graph.
    To do so, we claim that the good edges in the neighbourhood of any vertex in $Q_K'$ must be isomorphic to a subgraph of the original admissible graph.\\
    
    Suppose for contradiction that this was false.
    This means that there are two good edges which share an endpoint $v$ in the query despite not sharing an endpoint in the original graph, and that in the query both good edges are involved in triangles with some vertex $u$.
    But consider the edge between $u$ and $v$ in the query. Since the two good edges do not share an endpoint in the original graph, only one of them can have been involved in a triangle with this edge in the original graph.
    Since both of them are involved in a triangle with this edge in the query, we therefore have that at least one of them is not good.\\

    Thus, we have shown that the good edges can be covered by $q/pn$ many induced subgraphs, each of which is a subgraph of the original graph.
    By \cref{lem:babypacking}, those subgraphs of the original graph can in total correspond to at most $10q^2p \log(n) + 10 q \log(n)$ distinct edges, and so there can be at most $10q^2p \log(n) + 10 q \log(n)$ distinct edge weights present. By assumption that the reduction is non-duplicitous, this tells us that there are at most $10q^2p \log(n) + 10 q \log(n)$ good edges. The remainder of the proof is identical to \cref{thm:specifiedparts-lb}.
\end{proof}

The modifications required to make our lower bounds for listing to all-parts detection and from all-parts detection to all-$IJ$ detection non-duplicitous are essentially identical, so we omit those proofs. We conclude by giving a proof of the combinatorial lemma.

\begin{proof}[Proof of \cref{lem:babypacking}]
    We will use a Chernoff-union type argument.
    However, we must be careful in choosing the correct structures to union bound over.
    Observe that, if a graph is $(a,b,c)$-squishable, then we can find an $(a,b,c)$-squishing with at least $b$ edges by first choosing $a$ \emph{induced} subgraphs of $G$, then placing those subgraphs onto $c$ vertices, and finally removing all ``forbidden'' edges (that is, all edges which would cause one of the subgraphs to fail to be induced in the packing).
    There are at most $\binom{n}{c}^a \cdot (c^c)^a \leq n^{2ca}$ ways to choose these induced subgraphs of $G$ and their mapping onto $c$ vertices, so if we can show that any such packing has at most an $n^{-3ca}$ probability of corresponding to $b$ distinct edges from our random graph, by a union bound we'll have the desired statement.\\

    For our parameter regime, this follows from a Chernoff bound. Once we've fixed the choices of vertex sets from $V(G)$ and their mappings onto $[c]$, there are now at most $c^2 = q^2$ distinct pairs of vertices in $G$ that get mapped to non-``forbidden'' locations in the packing. 
    Each of those pairs of vertices becomes an edge in $G$ with independent probability $p$.
    So, the probability of seeing more than $10q^2p \log(n) + 10 q \log(n)$ of them become edges is at most 
    \begin{align*}
        \Pr[\text{Bin}(q^2, p) > 10q^2p \log(n) + 10 q \log(n)] < n^{-3 q^2 p - q} < n^{-3ca}.
    \end{align*}
\end{proof}

\section{Acknowledgments}

We would like to thank Alek Westover for several early conversations about the algorithms, and Carl Schildkraut for discussion of potential strengthenings to the results of \cref{sec:duplicitousness}.
\printbibliography

\end{document}